\documentclass[10pt, onecolumn, a4paper]{article}

\usepackage[utf8x]{inputenc} 
\usepackage{amsmath}
\usepackage{amsthm}
\usepackage{mathrsfs}
\usepackage{amsfonts}
\usepackage{amssymb}
\usepackage[dvips]{graphicx}
\usepackage{graphicx}
\usepackage{titletoc}
\usepackage{lmodern}
\usepackage[T1]{fontenc}
\usepackage{textcomp}
\usepackage{tikz}
\usepackage{color}
\usepackage{caption}
\usepackage{multicol}
\usepackage{parskip}
\usepackage{wrapfig}
\newtheorem{assumption}{Assumption}
\newtheorem{theorem}{Theorem}
\newtheorem{proposition}{Proposition}

\newtheorem{lemma}{Lemma}
\newtheorem{remark}{Remark}

\def\rea{\mathbb{R}}

\newcommand{\diag}{\mbox{\tt diag}}

\begin{document}
	\title{{Extended balancing of continuous LTI systems: a structure-preserving approach}}
\date{}
\author{Pablo Borja$^{1}$, Jacquelien M.A. Scherpen$^{1}$, and Kenji Fujimoto$^{2}$\\[0.1cm]\small
 $^{1}$Jan C. Willems Center for Systems and Control, ENTEG, FSE,\\\small 
        %Nijenborgh 4, 9747 AG Groningen, 
        University of Groningen. The Netherlands.\\\small Email:{\tt l.p.borja.rosales[j.m.a.scherpen]@rug.nl}\\[0.05cm]\small
$^{2}$ Dept of Aeronautics and Astronautics,\\\small Kyoto University, Japan.\\\small Email:{\tt fujimoto@kuaero.kyoto-u.ac.jp.}}
\maketitle
%
%%%%%%%%%%%%%%%%%%%%%%%%%%%%%%%%%%%%%%%%%%%%%%%%%%%%%%%%%%%%%
\begin{abstract}
In this paper, we treat extended balancing for continuous-time linear time-invariant systems, and we address the problem of structure-preserving model reduction of the subclass of port-Hamiltonian systems. We establish sufficient conditions to ensure that the reduced-order model preserves a port-Hamiltonian 
structure. Moreover, we show that the use of extended Gramians can be exploited to get a small error bound and, possibly, to preserve a physical interpretation for the reduced-order model.

\end{abstract}
%%%%%%%%%%%%%%%%%%%%%%%%%%%%%%%%%
%
\noindent {\bf Keywords:}
port-Hamiltonian systems, model reduction, extended Gramians, error bound.
%
%%%%%%%%%%%%%%%%%%%%%%%%%%%
\section{Introduction}
\label{sec:int}
%%%%%%%%%%%%%%%%%%%%%%%%%%%%

% High-order mathematical models have become more and more common in different areas of engineering, this phenomenon has different reasons, e.g., 
% the development of new technologies, the high scale integration of electrical circuits, the use of networks theory to study complex processes, 
% etc. Nonetheless, the manipulation of these high-order models can be a daunting task, and therefore, a lower-order approximation is necessary
% for different purposes, including simulation and control design. Furthermore, when the original model represents a physical system, it might be
% highly desirably to obtain a reduced order model that has a physical interpretation or, at least, preserves some appealing features of the original model, e.g., stability,
% passivity or even a particular structure.

Balancing is a tool that is often used for model reduction purposes, giving rise to the balanced truncation methodology. This approach  
relies on realization theory, observability and 
controllability Gramians and is directly related to the concept of Hankel operator of a system. Moreover, since its introduction in the seminal work of Moore \cite{MOORE}, balancing for stable linear systems has been 
extensively studied, in particular, a thorough exposition of this topic can be found in \cite{ANT}, while in \cite{SHandbook11} a brief tutorial is presented, which provides the 
basis for extending the results to nonlinear systems \cite{FS10}.

Balanced truncation, based on the use of standard observability and controllability Gramians, preserves some appealing properties of the original system, e.g., asymptotic stability, observability and controllability.
Furthermore, it is possible to establish an error bound, which is given in terms of the so-called Hankel singular values \cite{GLO} corresponding to the truncated states. Nevertheless, in this
standard formulation of balanced truncation some properties of the original system, like passivity or particular structures, are not necessarily preserved. Another possible drawback of this approach takes place when  
the Hankel singular values are large, which gives origin to a large error bound. Accordingly, with the aim of dealing with the latter issue, the use of the so-called generalized 
Gramians for model reduction purposes was introduced in \cite{HP90}. Where the generalized observability and controllability Gramians are solutions to the respective Lyapunov inequalities,
this differs from the definition of the standard Gramians which are given by the solutions of the Lyapunov equalities. 
Furthermore, it has been proven that it is possible to preserve some important properties while using balanced truncation based on the use 
of generalized Gramians. Moreover, since the solutions of the before mentioned Lyapunov inequalities are 
not unique, generalized Gramians can be used to obtain a smaller error bound \cite{DULL}, and in some cases, to preserve some interesting structures \cite{CS17}.

A further extension of balanced truncation can be formulated by using the concept of extended Gramians, which, for the discrete-time versions were introduced in \cite{SAND10}; and a preliminary continuous-time
counter part of these results was recently reported in \cite{SF18}. The discrete-time and continuous-time methods are rather different, except from the fact that the disspativity theory plays a fundamental role in both to establish the error bound. 
In this approach, referred as extended balancing, the Gramians are solutions to specific linear matrix inequalities 
(LMIs) and, in contrast to other balancing methods, the error bound is obtained by using dissipativity arguments \cite{W02} and not through a transfer function approach. Furthermore,  
this balancing method provides more degrees of freedom to impose certain structure to the reduced order model, 
and can be potentially useful to improve the error bound.

In this work, we focus on the extended balanced truncation of continuous-time linear time-invariant (CTLTI) systems, where we are interested in the versatility
of this methodology to preserve particular structures. Notably, we pay special attention to CTLTI port-Hamiltonian (PH) systems which
are suitable to represent several physical systems, e.g., RLC circuits and mechanical systems; and are endowed with interesting properties, such as passivity. 
Therefore, the objective of this work is not only to reduce the order of the original system, but also to preserve its PH structure. 
Towards this end, we first study extended balanced truncation for CTLTI systems, and then we focus on its application to CTLTI PH with structure preservation purposes. 
The main contributions of this paper are given as follows:
\begin{itemize}
\item We recall the results from \cite{SF18}, and provide proofs for the error bound computation which turn out to be rather different than in the dicrete-time case.
 \item We identify a family of generalized Gramians that are suitable for balanced truncation of CTLTI PH systems with PH structure preservation. To the best of our knowledge, 
 the characterization of these solutions to the Lyapunov inequalities is new.
 \item The use of extended balancing as a tool to design a small error bound. Moreover, we show with an illustrative example
 that this approach can be used to preserve more particular structures, like RLC circuits structure, and a physical interpretation for the reduced order model. 
\end{itemize}

The remainder of the paper is organized in the following manner: we provide the basic background in Section \ref{sec:ppf}, while the fundamental notion of extended Gramians 
and the computation of the error bound are presented in Section \ref{sec:eg}. In Section \ref{sec:bfphs}, we introduce the generalized and extended balancing of PH systems with structure preservation.
We present two illustrative examples in Section \ref{sec:ex},  where the use of extended Gramians in the second example allow us to preserve an even more particular structure than the PH one,
that is, the reduced order system is physically interpretable as an RLC circuit again.
Finally, in Section \ref{sec:cr} we wrap-up this note with some concluding remarks.
\\[.2cm]
{\bf Notation:} We assume that all the matrices have exclusively real entries. Consider a symmetric matrix $A\in \rea^{n \times n}$, then $A$ is positive semi-definite if $x^{\top}Ax\geq 0,\; \forall x\in\rea^{n}$. 
Moreover, $A$ is positive definite if $x^{\top}Ax> 0,\; \forall x\in\rea^{n}\backslash\{0\}$. The identity matrix is denoted as $I$, when necessary a subscript is added to indicate the dimension of the matrix. The symbol $\mathbf{0}_{q\times p}$ denotes a matrix of 
dimensions $q \times p$
whose entries are zeroes
The set of positive real numbers is expressed as 
$\rea_{>0}$, while, the set of nonnegative real numbers is denoted by $\rea_{\geq0}$. In the sequel, the symbol $\Lambda$ is reserved for diagonal matrices with positive entries, that is, 
the square matrix $\Lambda\in \rea^{n\times n}$ is given by $\Lambda=\diag\{\sigma_{1}, \cdots, \sigma_{n} \}$, where $\sigma_{i}\in \rea_{>0}$, for $i=1, \cdots n$. The symbol $U$ is reserved to orthogonal matrices, that is, 
$UU^{\top}=I$. Consider the vector $x\in\rea^{n}$, then $\lvert x \rvert$ denotes the Euclidean norm of $x$, that is, $\lvert x \rvert = \sqrt{x^{\top}x}$. Let $e\in \rea^{n}$ be a signal, then 
$\lVert e \rVert_{2}$ denotes the $\mathcal L_{2}$ norm of $e$, namely, $\lVert e \rVert_{2} = \left(\displaystyle\int_{0}^{\infty} \lvert e(t) \rvert^{2} dt\right)^{\frac{1}{2}}$.

%%%%%%%%%%%%%%%%%%%%%%%%%%%%
\section{Preliminaries}
\label{sec:ppf}
%%%%%%%%%%%%%%%%%%%%%%%%%%%%

Consider a continuous-time linear time-invariant (CTLTI) system described as
\begin{equation}
\mathcal{\varSigma}:\left\lbrace\begin{array}{rcl}
\dot{x}&=& Ax+Bu \\[.07cm] y&=& Cx,        
       \end{array}\right.\label{LTI}
\end{equation}
where $x\in \rea^{n}$ is the state vector, for $m\leq n$, $u\in \rea^{m}$ is the input vector and $y\in \rea^{q}$ denotes the output vector. Accordingly, $A\in \rea^{n \times n}$,
$B\in \rea^{n\times m}$ and $C\in \rea^{q\times n}$.
Assume that the system \eqref{LTI} is {\it asymptotically} stable, thus, the so-called \textit{generalized observability Gramians} $Q\in \rea^{n\times n}$ are positive semi-definite solutions to the following Lyapunov inequality
\begin{equation}
 QA+A^{\top}Q+C^{\top}C\leq 0 \label{go}.
\end{equation}
Analogously, the \textit{generalized controllability Gramians} $\breve P\in \rea^{n\times n}$ are given by positive semi-definite solutions to 
\begin{equation}
 A\breve P+\breve PA^{\top}+BB^{\top}\leq0. \label{gc}
\end{equation}

In particular, when \eqref{go} and \eqref{gc} are equalities, the matrices $Q$ and $\breve P$ are known as the \textit{standard} observability and controllability Gramian, respectively.
For further details, we refer the reader to \cite{ANT}.
%%%%%%%%%%%%%%%%%%%%%%%%%%%%%%%%%%%%%%%%%%%
\subsection{Generalized balanced truncation for LTI}\label{sec:GBT}

A CTLTI system is said to be \textit{generalized balanced} if 
\begin{equation}
 Q=\breve P=\Lambda_{QP},
\end{equation} 
where $\Lambda_{QP}>0$ is a diagonal matrix, see the Notation section.
Accordingly, balancing for LTI systems, \cite{MOORE}, relies on obtaining an invertible state transformation 
\begin{equation}
 \bar{x}=W^{-1}_{g}x \label{tr}
\end{equation} 
such that
\begin{equation}
 W^{-1}_{g}\breve PQW_{g}=\Lambda_{QP}^{2}, \label{balPQ}
\end{equation} 
where we assume that the elements of $\Lambda_{QP}=\diag\{\sigma_{1}, \cdots, \sigma_{n}\}$ are ordered from largest to smallest, that is, $\sigma_{i}>\sigma_{i+1}$, for $i=1,\cdots, n-1$.
Model reduction based on balancing is carried out by truncating the states corresponding to the small elements of $\Lambda_{QP}$, i.e., 
if $\sigma_i >> \sigma_{i+1}$, then we set
\begin{equation}
 \bar x_{i+1}=\cdots = \bar x_n=0.
\end{equation} 
The error bound is given by the sum of the truncated singular values \cite{GLO}, i.e.,
\begin{equation}
 \lVert \varSigma - \widehat{\varSigma}\rVert_\infty \leq 2 \sum_{j=i+1}^n \sigma_j, \label{generror}
\end{equation}
where $\widehat{\varSigma}$ corresponds to the realization of the reduced order system. 
For a more elaborated exposition of balancing and the corresponding reduced order model properties, we refer the reader to \cite{ZDG96}. At this point, we highlight that 
the error bound obtained through generalized balanced truncation is lower than the one obtained with the use of standard Gramians, for further details see \cite{HP90}.
% 
% Accordingly, in balanced truncation we look for a full rank matrix $W\in \rea^{n \times n}$ such that
% 
% Moreover, if we define the new coordinates $\bar{x}=Wx$ we get
% \begin{equation}
%  \begin{array}{rcl}
%   \dot{\bar x} &=& \underbrace{WAW^{-1}}_{\bar A}\bar{x}+\underbrace{WB}_{\bar B}u \\
%   y &=& CW^{-1}\bar{x}
%  \end{array}
% \end{equation} 

%%%%%%%%%%%%%%%%%%%%%%%%%%%
\section{Extended balanced truncation}
\label{sec:eg}
%%%%%%%%%%%%%%%%%%%%%%%%%%%

The generalized balanced truncation approach can be extended by considering the so-called \textit{extended Gramians} instead of the generalized ones.
This extension has two main advantages: on one hand, the error bound can be reduced as has been shown in \cite{SAND08} for the discrete-time case. On the other hand, the use of 
extended Gramians provides extra degrees of freedom which can be exploited to impose a certain structure on the reduced order system.

In this section we revisit and significantly improve the concept of extended balanced truncation for the continuous-time case, which was first introduced in \cite{SF18}. Towards this end, we 
introduce the following assumption which is necessary to establish the concept of extended Gramians.\\[.1cm]
\begin{assumption}\label{A1}\em
There exist strictly positive solutions, $Q, \ \breve P$, to inequalities \eqref{go} and \eqref{gc}.
\end{assumption}

We stress the fact that if the system \eqref{LTI} is controllable and observable, then Assumption \ref{A1} holds. Nonetheless, this latter condition is sufficient but not necessary, thus, 
might be conservative. Moreover, if Assumption \ref{A1} is satisfied, then we can define
\begin{equation}
 P:=\breve P^{-1}. \label{P2Pt}
\end{equation} 
Note that $P$ is a positive definite matrix.\\[0.1cm]
Before introducing the concept of extended Gramians we define the following matrices
\begin{equation*}
 \begin{array}{rcl}
  A_{o}&:=& \alpha I_{n}+A,\\
  A_{c}&:=& \beta I_{n}+A, \\
  X_{o}&:=&-QA-A^{\top}Q-C^{\top}C,\\
  X_{c}&:=&-PA-A^{\top}P-PBB^{\top}P,\\
  Y_{c}&:=&-P+(A_{c}^{\top}+PBB^{\top})T,
 \end{array}
\end{equation*} 
where $P$ is defined in \eqref{P2Pt}, $\alpha\in \rea_{>0}$, and $\beta\in \rea_{\geq0}$. Furthermore, from \eqref{go} and \eqref{gc}, it follows that $X_{o}\geq0$, $X_{c}\geq0$.\\[0.1cm]
The definition of extended Gramians is the starting point of the theory contained in the following sections of this paper. These concepts were introduced for CTLTI systems without proofs in
\cite{SF18}. Below we present the, slightly altered, results and their corresponding proof.\\[.1cm]
%%%%%%%%%%%%%%
%%%%%%%%%%%%%%
\textit{Extended Gramians.} Consider the following two LMIs.
\begin{equation}
\begin{bmatrix} 
X_{o} & Q -A_{o}^{\top}S \\[0.1cm] Q - S^\top A_{o}& S + S^\top 
\end{bmatrix} \geq 0
\label{LMIobs}
\end{equation}
and 
\begin{equation}
\begin{bmatrix} 
-PA-A^\top P & -P+A_{c}^{\top} T & -2PB  \\[0.1cm] -P +T^\top A_{c} & T + T^\top & 2 T^\top B \\[0.1cm] -2B^\top P & 2B^\top T & 4I_{m} 
\end{bmatrix} \geq 0 \label{LMIcontr}
\end{equation}
with $T,S \in \mathbb{R}^{n \times n}$. We call (\ref{LMIobs}) and (\ref{LMIcontr}) the extended observability and controllability LMIs with extended observability Gramian $(Q, S, \alpha)$ 
and extended inverse controllability 
Gramian $(P, T, \beta)$, respectively. \\[.2cm]
%%%%%%%%%%%%%%%%%%
%%%%%%%%%%%%%%%%%%
Now we are in position to formulate the relation between the generalized observability Gramian and the extended observability Gramian.\\
%%%%%%%%%%%%%%%%%%%
\begin{theorem} {\it (observability Gramians)}
\label{tm:obs}\\
The inequality \eqref{go} has a solution $Q>0$ if and only if the LMI \eqref{LMIobs} admits a solution $(Q,S,\alpha)$ with $Q > 0$, $(S+S^{\top})\geq 0$, and $\alpha$ large enough. 
Moreover, if $X_{o}>0$, then there exists an $\alpha$ large enough, and $S=S^{\top}>0$ such that the LMI \eqref{LMIobs} holds.
\end{theorem}
%%%%%%%%%%%%%%%%%%%
\begin{proof}
{\it Only if.} Assume that \eqref{LMIobs} has a solution $(Q,S,\alpha)$, then multiplying \eqref{LMIobs} by $[I_n \; \; \mathbf{0}_{n\times n}]$ from the left and 
by $[I_n \; \; \mathbf{0}_{n\times n}]^{\top}$ from the right, it follows that \eqref{go} admits a solution $Q > 0$.\\
{\it If.} Assume there exists $Q > 0$ solving \eqref{go}. Select $S= A^{-\top}_{o}Q$, with $-\alpha$ not an eigenvalue of $A$. Then, the off-diagonal blocks of 
\eqref{LMIobs} are zero. Furthermore, 
\begin{equation}
 S+S^{\top}=A_{o}^{-\top}Q + QA_{o}^{-1}.
\end{equation} 
Accordingly, we have the following equivalence
\begin{equation}
 \begin{array}{l}
  0\leq S+S^{\top}\Longleftrightarrow \\[.1cm]
  \begin{array}{lcl}
   0\leq A_{o}^{\top}(S+S^{\top})A_{o}&=&A_{o}^{\top}Q+QA_{o}\\
  &=&2\alpha Q - C^{\top} C-X_{o}.
  \end{array}
 \end{array}\label{poscond}
\end{equation} 
Note that, since $X_{o}$ does not depend on $\alpha$, the inequality \eqref{poscond} holds for $\alpha$ large enough. Hence, there exist $Q>0$ and $\alpha>0$ such that LMI \eqref{LMIobs} holds.\\
\textit{Symmetric $S$.} Assume that $Q>0$ and $X_{o}>0$. Consider a symmetric matrix $\Gamma_{o}\in \rea^{n\times n}$
verifying
\begin{equation}
 \alpha Q+\Gamma_{o}>0. \label{posQ}
\end{equation} 
Select
\begin{equation}
 S=Q\left(\alpha Q+\Gamma_{o}\right)^{-1}Q. \label{symS}
\end{equation} 
Hence, $S=S^{\top}>0$. Now,
multiply \eqref{LMIobs} by $\diag\{ I_{n},QS^{-1}\}$ from the left and by $\diag\{ I_{n},S^{-1}Q\}$ from the right, yielding 
 \begin{equation}
 \begin{array}{l}
 \begin{bmatrix}
   X_{o} & QS^{-1}Q-A_{o}^{\top}Q \\[.1cm] QS^{-1}Q-QA_{o} & 2QS^{-1}Q 
 \end{bmatrix}\\[.5cm] \; \; \; \; \; \; \; \; \; \; \; \; \; \; \; \; \; \; \; \; \; \; \; \;= \begin{bmatrix}
                            X_{o} & \Gamma_{o}-A^{\top}Q \\ \Gamma_{o}-QA & 2(\alpha Q+\Gamma_o)  
                            \end{bmatrix}
 \geq0.
\end{array}
 \label{eot}
 \end{equation}
Furthermore, LMI \eqref{eot} is equivalent through Schur complement to
\begin{equation}
 2\alpha Q+2 \Gamma_{o}-\Theta_{o}\geq 0, 
\label{obsym}
\end{equation} 
with
\begin{equation*}
 \Theta_{o}:=(\Gamma_{o} - QA)X_{o}^{-1}(\Gamma_{o} -A^{\top}Q).
\end{equation*} 
Note that there exists $\alpha$, large enough, such that \eqref{obsym} is satisfied. This completes the proof.
 \end{proof}
The results on generalized and extended observability Gramians have a controllability version as follows.\\
%%%%%%%%%%%%%%%%%%%%
\begin{theorem}
{\it (controllability Gramians)}
\label{tm:contr}\\
The inequality \eqref{gc} has a solution $\breve P>0$ if and only if the LMI \eqref{LMIcontr} has a solution 
$(P,T,\beta)$ with $P > 0$. Furthermore, if $X_{c}>0$, then there exists a $\beta>0$ large enough, and $T=T^{\top}>0$ such that 
the LMI \eqref{LMIcontr} holds.\vspace{0.1cm}
\end{theorem}
%%%%%%%%%%%%%%%%%%%%
\begin{proof}
In order to establish the proof, note that \eqref{LMIcontr} is equivalent to the following LMI
\begin{equation}
\begin{bmatrix}
   X_{c} &  Y_{c}\\ Y_{c}^{\top} & T+T^{\top}-T^{\top}BB^{\top}T
\end{bmatrix}\geq0. \label{ect}
 \end{equation}
 {\it Only if.} Assume that \eqref{LMIcontr} admits a solution $(P,T,\beta)$ with $P>0$, thus equivalently, \eqref{ect} is satisfied. 
 Multiplying the latter LMI by $[I_n \; \; \mathbf{0}_{n\times n}]$ from the left and 
by $[I_n \; \; \mathbf{0}_{n\times n}]^{\top}$ from the right, it follows that 

\begin{equation}
 \begin{array}{rcl}
  X_{c}&\geq& 0 \\[0.1cm]
 \Longleftrightarrow -PA-A^{\top}P-PBB^{\top}P&\geq& 0,\\[0.1cm]
 \Longleftrightarrow A\breve P+\breve PA^{\top}+BB^{\top}&\leq&0,
 \end{array} 
\end{equation} 
where we used \eqref{P2Pt} to obtain the last inequality.\\
{\it If.} Assume there exists $\breve P > 0$ solution to \eqref{gc}.
Fix\footnote{Since $\beta\geq0$ and $\Re\{ \lambda(A) \}<0$, $\beta$ is not an eigenvalue of $A$.} $T= P(\beta I_{n}-A)^{-1}$, with $P$ defined in \eqref{P2Pt}, then we get

\begin{equation}
 \begin{array}{rcl}
  Y_{c}&=&-P+(A_{c}^{\top}+PBB^{\top})P(\beta I_{n}-A)^{-1}\\[0.1cm]
%  &=& -P+(\beta P + A^{\top}P+PBB^{\top}P)(\beta I_{n}-A)^{-1} \\
  &=& -P + (\beta P - PA - X_{c})(\beta I_{n}-A)^{-1}\\[0.1cm]
%  &=& -P + P (\beta I_{n}-A-X_{c}P) (\beta I_{n}-A)^{-1}\\
  &=& -X_{c}(\beta I_{n}-A)^{-1}\\[0.1cm]
  &=& -X_{c}\breve PT,
 \end{array}
\end{equation} 

and

\begin{equation}
 \begin{array}{rcl}
  T+T^{\top}-T^{\top}BB^{\top}T&=& T^{\top}\left(T^{-1}+T^{-\top}-BB^{\top}\right)T\\[0.1cm]
  %&=&T^{\top}\left[(\beta I_{n}-A)\breve{P}+\breve{P}(\beta I_{n}-A^{\top})  \right]T\\
  &=&T^{\top}\breve P\left(2\beta P+X_{c}\right)\breve PT.
 \end{array}
\end{equation}

Hence, LMI \eqref{ect} takes the form 

\begin{equation}
\begin{bmatrix}
 X_{c} & -X_{c}\breve PT \\ T^{\top}\breve PX_{c} & T^{\top}\breve P\left(2\beta P+X_{c}\right)\breve PT
\end{bmatrix}\geq0.
\label{ect1}
\end{equation} 

Now, we multiply \eqref{ect1} by $\diag\{I_{n}, PT^{-\top} \}$ from the left, and by $\diag\{I_{n}, T^{-1}P \}$, yielding 

\begin{equation}
 \begin{bmatrix}
  X_{c} & -X_{c} \\ -X_{c} & X_{c}  
 \end{bmatrix}
+\begin{bmatrix}
  \mathbf{0}_{n\times n} & \mathbf{0}_{n\times n} \\ \mathbf{0}_{n\times n} & 2\beta P
 \end{bmatrix}\geq0
\end{equation} 
which holds for every $\beta\geq 0$.\\
\textit{Symmetric $T$.} Assume that $P>0$ and $X_{c}>0$. 
Consider a symmetric matrix $\Gamma_{c}\in \rea^{n\times n}$
verifying
\begin{equation}
 \beta \breve{P}+\Gamma_{c}>0. \label{posT}
\end{equation} 
Select
\begin{equation}
 T=\left(\beta \breve{P}+\Gamma_{c}\right)^{-1}. \label{Tsel}
\end{equation} 
Hence, $T=T^{\top}>0$.
Multiply \eqref{ect} by $\diag\{I_{n},T^{-\top}\}$ from the left and by
$\diag\{I_{n},T^{-1}\}$ from the right, and substitute \eqref{Tsel} to obtain

\begin{equation}
  \begin{bmatrix}
   X_{c} & -P\Gamma_{c}+A^{\top}+PBB^{\top} \\ -\Gamma_{c}P+A+BB^{\top}P & 2(\beta \breve{P}+\Gamma_{c})-BB^{\top}
  \end{bmatrix}\geq0,
\label{ect2}
\end{equation}
which is equivalent to 
\begin{equation}
 2\beta \breve{P}+2\Gamma_{c}-BB^{\top}-\Theta_{c}\geq0
\end{equation} 
where
\begin{equation*}
  \Theta_{c}:=(-\Gamma_{c}P+A+BB^{\top}P)X_{c}^{-1}(-P\Gamma_{c}+A^{\top}+PBB^{\top}).
 \end{equation*}

 Since $\Theta_{c}$ does not depend on $\beta$, it follows that LMI \eqref{ect2}, and in consequence LMI \eqref{LMIcontr}, holds for $\beta>0$ large enough.
 This completes the proof.
\end{proof}
\begin{remark}
For clarity of presentation, we assume that $X_{o}>0,\ X_{c}>0$ to prove the existence of symmetric solutions to \eqref{LMIobs} and \eqref{LMIcontr}, respectively.
While these conditions are not restrictive, they can be relaxed to $X_{o}\geq0,\ X_{c}\geq0$ by using generalized inverses. This however needs the introduction of the following conditions
\begin{equation}
 \begin{array}{rcl}
  (I_{n}-X_{o}X_{o}^{\dagger})(\Gamma_{o}-A^{\top}Q)&=&\mathbf{0}_{n\times n}\\[0.1cm]
  (I_{n}-X_{c}X_{c}^{\dagger})(-P\Gamma_{c}+A^{\top}+PBB^{\top})&=&\mathbf{0}_{n\times n},
 \end{array}\label{pseudo}
\end{equation} 
where $X_{o}^{\dagger},\ X_{c}^{\dagger}$ denote generalized inverses of $X_{o}$ and $X_{c}$, respectively. Note that both expressions in \eqref{pseudo} are naturally satisfied 
if $X_{o}>0,\ X_{c}>0$.\\
\end{remark}
\begin{remark}
 The symmetric matrices $\Gamma_{o}$ and $\Gamma_{c}$ provide degrees of freedom in the selection of the extended Gramians. These degrees of freedom can be used to improve the error bound in case the Gramians are used for model reduction, see Section \ref{sec:error},
 or to impose a desired structure to the reduced order model as is illustrated in Section \ref{sec:ex}. 
\end{remark}
For the model reduction application, we assume that the matrices $S$ and $T$ are symmetric. From Theorems \ref{tm:obs} and \ref{tm:contr}, it is clear that this assumption is not necessary to ensure the existence of solutions to 
\eqref{LMIobs} and \eqref{LMIcontr}, but we need it for obtaining an error bound in Section \ref{sec:error}.

In the extended balancing approach, a CTLTI system is said to be \textit{extended balanced} if
\begin{equation*}
 S=T^{-1}=\Lambda_{ST},
\end{equation*} 
where $\Lambda_{ST}$ is a diagonal matrix, see the Notation section.
Therefore, we look for an invertible state transformation 
\begin{equation}
 \bar{x}=W_{e}^{-1}x \label{tre}
\end{equation} 
such that
\begin{equation}
 W_{e}^{-1}T^{-1}SW_{e}=\Lambda_{ST}^{2}. \label{balST}
\end{equation} 
Similar to Section \ref{sec:GBT}, we assume that the elements of the diagonal matrix $\Lambda_{ST}$ are ordered from largest to smallest. 
Hence, the order of the CTLTI system is reduced by truncating the states that correspond to the smallest elements of the aforementioned matrix.

The discrete-time version of the LMIs \eqref{LMIobs} and \eqref{LMIcontr} can be found in \cite{OLI99} and \cite{OLI02}. 
While, a thorough exposition of extended balanced truncation for discrete-time
linear time-invariant (DTLTI) systems is given in \cite{SAND08} and \cite{SAND10}.
%%%%%%%%%%%%%%%%%%%%%%%%%%%%%%%%%
%%%%%%%%%%%%%%%%%%%%%%%%%%%%%%%%%
%%%%%%%%%%%%%%%%%%%%%%%%%%%%%%%%%
\subsection{Computation of the error bound}\label{sec:error}
%%%%%%%%%%%%%%%%%%%%%%%%%%%%%%%%%
%%%%%%%%%%%%%%%%%%%%%%%%%%%%%%%%%
%%%%%%%%%%%%%%%%%%%%%%%%%%%%%%%%%
One of the appealing features of the balanced truncation approach is the possibility of establishing a clear error bound. For the generalized balanced truncation 
case, the inequality \eqref{generror} establishes the error bound, which is customarily obtained through the analysis in the frequency domain of the original system and the reduced order one \cite{GLO}, \cite{ZDG96}.
In this subsection we provide a procedure, different from the approach proposed in \cite{SF18}, to compute such error bound for the extended case. Towards this end, we assume that the linear 
transformation $W_{e}$, such that \eqref{balST} holds, is known. Then, we introduce the following state-space systems
\begin{eqnarray}
  \bar{\varSigma}&:&\left\lbrace
 \begin{array}{rcl}
\dot{\bar{x}}&=& \bar A \bar x+\bar Bu \\[.07cm] \bar y&=& \bar C\bar x,
       \end{array}\right.\label{sysbal} \\[0.2cm]
       \varSigma_{r}&:&\left\lbrace\begin{array}{rcl}
\dot{x}_{r}&=& \bar A x_{r}+\bar Bu+v(t) \\[.07cm] y_{r}&=& \bar C x_{r},
       \end{array}\right. \label{sysrred} 
\end{eqnarray}
where $\bar x$ is defined as in \eqref{tre}, $v(t)\in\rea^{n}$ is an external signal, $x_{r}\in\rea^{n}$ is an auxiliary state, and
\begin{equation}
 \begin{array}{rcl}
  \bar A:=W_{e}^{-1}AW_{e}, &
  \bar B:= W_{e}^{-1}B, &
  \bar C:= CW_{e}.
 \end{array}\label{balpar}
\end{equation} 
%Note that the systems $\bar{\varSigma}$ and $\varSigma_{r}$ are identical, except from $v(t)$ and, potentially, the states initial conditions.\\[0.1cm] 
Now, we  split $\bar x$
system into two parts, namely, 
\begin{equation}
 \bar x = \begin{bmatrix}
           \bar x_{1} \\ \bar x_{2}
          \end{bmatrix}, \label{split}
\end{equation} 
where $\bar x_{1}\in \rea^{k}$ is the part of the state to be preserved after the 
reduction of the model and $\bar x_{2}\in \rea^\ell$, with $\ell:=n-k$, is the part to be truncated. Accordingly, the matrices given in \eqref{balpar}
can be expressed as follows
\begin{equation*}
 \begin{array}{rcl}
  \bar A=\begin{bmatrix}
         \bar{A}_{11} & \bar{A}_{12} \\ \bar{A}_{21} & \bar{A}_{22}    
            \end{bmatrix},
 &
  \bar B= \begin{bmatrix}
              \bar{B}_{1} \\ \bar{B}_{1}
             \end{bmatrix},
 &
  \bar C= \begin{bmatrix} \bar{C}_{1} \\ \bar{C}_{2}
                   \end{bmatrix},
 \end{array}
\end{equation*}
with 
\begin{equation*}
 \begin{array}{llll}
  \bar{A}_{11}\in\rea^{k\times k}, & \bar{A}_{12}\in \rea^{k\times \ell}, & \bar{A}_{21}\in \rea^{\ell\times k}, &  \bar{A}_{22}\in \rea^{\ell\times\ell}, \\[0.1cm]
  \bar{B}_{1}\in \rea^{k\times m}, & \bar{B}_{2}\in \rea^{\ell\times m}, & \bar{C}_{1}\in \rea^{q\times k}, & \bar{C}_{2}\in \rea^{q\times \ell}. 
 \end{array}
\end{equation*} 
Thus, the truncation of the state $\bar{x}_{2}$ leads to the following reduced order model
\begin{equation}
 \widehat\varSigma:\left\lbrace\begin{array}{rcl}
\dot{\hat x}&=& \widehat A \hat x+\widehat Bu \\[.07cm] \hat y&=&\widehat C\hat x,        
       \end{array}\right. \label{sysr} 
\end{equation}
where 
\begin{equation*}
 \begin{array}{cccc}
 \hat{x}=\bar{x}_{1},&
\widehat A:=\bar A_{11}, &
\widehat B:= \bar{B}_{1}, &
\widehat C:= \bar{C}_{1}.
 \end{array}
\end{equation*} 
Now, inspired by the ideas presented in \cite{W02}, and by the approach adopted in \cite{SAND08}, \cite{SAND10} for discrete-time, and in \cite{SF18} for continuous-time, 
we propose a storage function that is instrumental to establish the error bound. Towards this end, we first introduce the following definitions to simplify the notation of this section: 
\begin{equation}
 \begin{array}{ll}
   \bar Q := W_{e}^{\top} Q W_{e} , & 
  \bar P := W_{e}^{\top} P W_{e}, \\[0.2cm]
  z_{o}:=\bar x - x_{r}, & z_{c}:=\bar x + x_{r}. 
%   \\[0.2cm]
%   \breve z_{o}:=\phi_{\bar Q}z_{o}, & \breve z_{c}:=\phi_{\bar P}^{-\top}z_{c},  \\[0.2cm]
%   \mathbf{A}_{o}:=\phi_{\bar Q}\bar A\phi_{\bar Q}^{-1}, & \mathbf{A}_{c}:=\phi_{\bar P}^{-\top}\bar A\phi_{\bar P}^{\top}, \\[0.2cm]
%   \mathbf{B}:=\phi_{\bar P}^{-\top}\bar B, & \mathbf{C}:=\bar C\phi_{\bar Q}^{-1}, 
 \end{array}\label{defs}
\end{equation} 
where $P$ is defined as in \eqref{P2Pt}. The proposition below introduces a storage function which is used to establish the error bound in this section.\\
%The following proposition introduces a storage function that is instrumental in the computation of the error bound.\\
%%%%%%%%%%%%%%%%%%%%%%%%%%%%
%%%%%%%%%%%%%%%%%%%%%%%%%%%%
%%%%%%%%%%%%%%%%%%%%%%%%%%%%
\begin{proposition}\label{prop:S}
 Consider the systems $\varSigma,\ \bar\varSigma, \ \varSigma_{r}$ given in \eqref{LTI}, \eqref{sysbal}, and \eqref{sysrred}, respectively.
 Assume that the triplet $(Q, S, \alpha)$ solves LMI \eqref{LMIobs} and the triplet $(P, T, \beta)$ solves LMI \eqref{LMIcontr}.
 Consider the storage function 
 \begin{equation}
 \mathcal{S}( z_{o},  z_{c})=z^{\top}_{o}\bar{Q}z_{o}  + \sigma_{n}^{2}z^{\top}_{c}\bar{P}z_{c} \label{Sto}
\end{equation} 
where $\sigma_{n}$ is the $n-th$ entry of $\Lambda_{ST}$, and ${z}_{o}, {z}_{c}$ are defined in \eqref{defs}. Then,
\begin{equation}
\begin{array}{rcl}
 \dot{\mathcal{S}} & \leq & 4\sigma_{n}^{2}\lvert u \rvert^{2} - \lvert  y - y_{r} \rvert^{2} \\&& + 
 2\left[\sigma_{n}^{2}\left( \beta  z_{c} + \dot{z}_{c}\right)^{\top} \Lambda_{ST}^{-1} - \left(\alpha z_{o} + \dot{z}_{o}\right)^{\top}\Lambda_{ST}\right] v\label{dS}
\end{array}
\end{equation} 
\end{proposition}
%%%%%%%%%%%%%%%%%%%%%%%%%%%%
\begin{proof}
 Note that
 \begin{equation}
  \dot{\mathcal{S}}= 2z^{\top}_{o}\bar{Q}\dot{z}_{o} + 2\sigma_{n}^{2}z^{\top}_{c}\bar{P}\dot{z}_{c}. \label{dS1}
 \end{equation} 
 Define the vectors
 \begin{equation}
  \begin{array}{rl}
   \xi_{o}:=
\begin{bmatrix}
             W_{e}z_{o} \\ W_{e}v
            \end{bmatrix},  &  \xi_{c}:=\begin{bmatrix}
             W_{e}z_{c} \\ W_{e}v \\u
            \end{bmatrix}.
  \end{array} 
 \end{equation} 
 Multiply LMI \eqref{LMIobs} by $\xi_{o}^{\top}$ from the left and by $\xi_{o}$ from the right, yielding
 \begin{equation}
\begin{array}{r}
   2\left[v^{\top}-z_{o}^{\top}(\alpha I_{n}+\bar A^{\top})\right]\Lambda_{ST}v \hfill
   \\[0.15cm]+z_{o}^{\top}\left[ W_{e}^{\top}X_{o}W_{e}z_{o}+2\bar Qv\right]\geq 0\\[0.15cm]
   \Longleftrightarrow-2(\dot z_{o}+\alpha z_{o})^{\top}\Lambda_{ST}v \hfill 
   \\[0.15cm]+z_{o}^{\top}\left[ W_{e}^{\top}X_{o}W_{e}z_{o}+2\bar Qv\right]\geq 0 \\[0.15cm]
   \Longleftrightarrow-2(\dot z_{o}+\alpha z_{o})^{\top}\Lambda_{ST}v - z_{o}^{\top}\bar{C}^{\top}\bar{C} z_{o}\hfill
   \\[0.15cm]+2{z}_{o}^{\top}\bar{Q}\left[v-\bar{A}{z}_{o}\right]\geq0 \\[0.15cm]
   \Longleftrightarrow-2(\dot z_{o}+\alpha z_{o})^{\top}\Lambda_{ST}v-\lvert y - y_{r} \rvert^{2}-2 z_{o}^{\top}\bar{Q}\dot{{z}}_{o}\geq0,
\end{array}\label{dQ1}
 \end{equation}
 where we used the facts
 \begin{equation}
  \begin{array}{rcl}
  \dot{{z}}_{o}&=& \bar{A}{z}_{o}- v \\
\bar{C}{z}_{o}&=&y-y_{r}.
\end{array}
 \end{equation} 
 Note that \eqref{dQ1} implies that
 \begin{equation}
   2 z_{o}^{\top}\bar{Q}\dot{{z}}_{o} \leq -2(\dot z_{o}+\alpha z_{o})^{\top}\Lambda_{ST}v-\lvert y - y_{r} \rvert^{2}. \label{dQ2}
 \end{equation} 
 Now, multiply LMI \eqref{LMIcontr} by $\xi_{c}^{\top}$ from the left and by $\xi_{c}$ from the right to obtain
 \begin{equation}
  \begin{array}{r}
   -2z_{c}^{\top}\bar P \left[ \bar A z_{c} + 2 \bar B u + v \right] + 4 \lvert u \rvert^{2} \hfill \\[0.15cm]+2\left[ z_{c}^{\top}(\beta I_{n} + \bar A^{\top}) + 
   v^{\top} +2u^{\top}\bar B^{\top}\right]\Lambda_{ST}^{-1}v \geq 0 \\[0.15cm]
   \Longleftrightarrow-2{z}_{c}^{\top}\left( \bar{A} {z}_{c} + 2 \bar{B}u + v \right) + 4 \lvert u \rvert^{2} \hfill \\[0.15cm]
   +2\left( \dot{z}_{c}+\beta z_{c}\right)^{\top}\Lambda_{ST}^{-1}v \geq 0 \\[0.15cm]
   \Longleftrightarrow  4 \lvert u \rvert^{2} + 2\left( \dot{z}_{c}+\beta z_{c}\right)^{\top}\Lambda_{ST}^{-1}v \geq 2 z_{c}^{\top}\bar{P}\dot{{z}}_{c},
  \end{array} \label{dP1}
 \end{equation} 
 where we used that
 \begin{equation}
  \dot{{z}}_{c}=\bar{A} {z}_{c} + 2 \bar{B}u + v.
 \end{equation} 
%  Inequality \eqref{dP1} implies
%  \begin{equation}
%    2\breve z_{c}^{\top}\dot{\breve{z}}_{c} \leq 4 \lvert u \rvert^{2} + 2\left( \dot{z}_{c}+\beta z_{c}\right)^{\top}\Lambda_{ST}^{-1}v \label{dP2}
%  \end{equation} 
 The proof is completed by substituting \eqref{dQ2} and \eqref{dP1} in \eqref{dS1} to obtain \eqref{dS}.
\end{proof}
%%%%%%%%%%%%%%%%%%%%%%%%%%%%%%%
%%%%%%%%%%%%%%%%%%%%%%%%%%%%%%%
%%%%%%%%%%%%%%%%%%%%%%%%%%%%%%%
In order to establish the error bound, we propose a particular 
selection of the signal $v(t)$ that allow us to compare the behavior of systems \eqref{sysbal} and \eqref{sysrred}.\\
%%%%%%%%%%%%%%%%%%%%%%%%%%%%%
%%%%%%%%%%%%%%%%%%%%%%%%%%%%%
%%%%%%%%%%%%%%%%%%%%%%%%%%%%%
\begin{lemma}\label{lemma1}
 Consider $\ell=1$. Assume that systems \eqref{sysrred} and \eqref{sysr} are initially at rest.
 Consider the partition $x_{r}=[x_{r_{1}}^{\top} \ x_{r_{2}}^{\top}]^{\top}$, with $x_{r_{1}}\in\rea^{n-1}$ and $x_{r_{2}}\in\rea$. Choose

\begin{equation}
  v(t)=
    -\begin{bmatrix} \mathbf{0}_{n-1} \\[0.1cm]
\bar A_{21}x_{r_{1}}(t)
+\bar{B}_{2}u(t)
\end{bmatrix}. \label{v}
\end{equation} 
Then, $\hat y(t)= y_{r}(t)$, and $x_{r_{2}}(t)=0$ for every $t\geq 0$.
\end{lemma}
%%%%%%%%%%%%%%%%%%%%%%
\begin{proof}
 To establish the proof replace \eqref{v} in \eqref{sysrred} to obtain
\begin{equation}
 \begin{array}{rcl}
    \dot{x}_{r_{1}}&=&\bar A_{11}x_{r_{1}}+\bar A_{12}x_{r_{2}}+\bar B_{1}u  \\
  \dot{x}_{r_{2}}&=&\bar A_{22}x_{r_{2}} 
 \end{array}\label{x2r}
\end{equation} 
Since $x_{r}(0)=\mathbf{0}_{n}$, from \eqref{x2r} we have the following chain of implications
\begin{equation}
 \begin{array}{rcl}
  \dot{x}_{r_{2}}=0 \; \; \forall \; \; t\geq 0 &\Longrightarrow& x_{r_{2}}(t) = 0 \; \; \forall \; \; t\geq 0 \\[0.15cm]
  &\Longrightarrow& \dot{x}_{r_{1}}=\bar A_{11}x_{r_{1}}+\bar B_{1}u. 
 \end{array}\label{x1r}
\end{equation} 
Since $\hat{x}(0)=\mathbf{0}_{n-1}$, the last expression of \eqref{x1r} implies that $\hat{x}(t)=x_{r_{1}}(t)$ for all $t\geq 0$. Hence,
\begin{equation}
  y_{r}= \bar C_{1}x_{r_{1}}
  = \widehat C\hat x
  = \hat{y}.
\end{equation} 
\end{proof}
%%%%%%%%%%%%%%%%%%%%%%%%%%%
%%%%%%%%%%%%%%%%%%%%%%%%%%
%%%%%%%%%%%%%%%%%%%%%%%%%%
Using the results of Proposition \ref{prop:S} and Lemma \ref{lemma1}, the following Lemma establishes an error bound for the case $\ell=1$, that is, when only one state is truncated.\\
%%%%%%%%%%%%%%%%%%%%%%%%%%%
%%%%%%%%%%%%%%%%%%%%%%%%%%
%%%%%%%%%%%%%%%%%%%%%%%%%%
\begin{lemma}\label{lemma2}
 Consider the balanced system \eqref{sysbal} with extended observability Gramian $(\bar Q, \Lambda_{ST}, \alpha)$, and inverse extended controllability 
 Gramian $(\bar P, \Lambda_{ST}^{-1}, \beta)$, where $\alpha=\beta$ and $\ell=1$.
 Assume that systems \eqref{bsys}, \eqref{sysr} and \eqref{sysrred} are initially at rest and select $v$ as in \eqref{v}. Then,
\begin{equation}
 \lVert \varSigma-\widehat{\varSigma} \rVert_{\infty}\leq 2\sigma_{n}.
\end{equation}
\end{lemma}
%%%%%%%%%%%%%%%%%%%%%%%%%
\begin{proof}
 Define
 \begin{equation}
  v_{2}:=\bar A_{21}x_{r_{1}}(t)
+\bar{B}_{2}u(t).
 \end{equation} 
 Hence, we can rewrite \eqref{v} as follows 
 \begin{equation}
 v= \begin{bmatrix}
   \mathbf{0}_{n-1} \\ v_{2}
  \end{bmatrix}. \label{v2}
 \end{equation} 
On the other hand, from Lemma \ref{lemma1} we have that
\begin{equation}
 \begin{array}{rl}
  x_{r}=\begin{bmatrix}
         \hat{x} \\ 0
        \end{bmatrix},
& y_{r}=\hat{y}.
 \end{array}
\end{equation} 
Therefore, since $\alpha=\beta$, we get
\begin{equation}
 \begin{array}{rcl}
  \left(\alpha z_{o} + \dot{z}_{o}\right)^{\top} \Lambda_{ST} v&=& \sigma_{n}(\alpha \bar x_{2}+\dot{\bar{x}}_{2})v_{2}\\[0.1cm]
  &=& \sigma_{n}^{2}\left( \beta z_{c} + \dot{z}_{c}\right)^{\top} \Lambda_{ST}^{-1} v.
 \end{array}\label{dSter}
\end{equation} 
Now, consider the storage function $\mathcal{S}( z_{o},  z_{c})$, given in \eqref{Sto}. Then, substituting \eqref{dSter} in \eqref{dS},
its derivative along the trajectories reduces to 
\begin{equation}
 \dot{\mathcal{S}}  \leq  4\sigma_{n}^{2}\lvert u \rvert^{2} - \lvert  y - \hat y \rvert^{2}, \label{dS2}
\end{equation} 
where we used \eqref{dSter}. Moreover, 
integrating \eqref{dS2} from $0$ to $\infty$, yields
\begin{equation}
 0\leq 4\sigma_{n}^{2}\lVert u \rVert^{2}_{2}-\lVert y- \hat y\rVert^{2}_{2}
\end{equation} 
which implies 
\begin{equation}
 \lVert y- \hat y\rVert_{2}\leq 2\sigma_{n}\lVert u \rVert_{2}.
\end{equation}
The proof is completed by using the induced $\mathcal{L}_{2}$ norm, see Proposition 5.13 and the table on page 150 of \cite{ANT}.
 \end{proof}
 Now we are in position to present the main result of this paper in terms of the error bound for model reduction for CTLTI systems based on extended balanced truncation.\\
%%%%%%%%%%%%%%%%%%%%%%
%%%%%%%%%%%%%%%%%%%%%%
%%%%%%%%%%%%%%%%%%%%%%
\begin{theorem}
 Consider the balanced system \eqref{sysbal} with extended observability Gramian $(\bar Q, \Lambda_{ST}, \alpha)$, and inverse extended controllability 
 Gramian $(\bar P, \Lambda_{ST}^{-1}, \beta)$, where $\alpha=\beta$ and
 \begin{equation*}
  \Lambda_{ST}=\diag\{\sigma_{1}, \cdots, \sigma_{n} \}.
 \end{equation*} 
 Consider the truncated $k^{th}$ order system \eqref{sysr}. Then,
 the error bound is given by the following inequality
 \begin{equation}
 \lVert \varSigma-\widehat{\varSigma} \rVert_{\infty}\leq 2\sum_{j=k+1}^{n}\sigma_{j}. \label{eb} 
\end{equation} 
\end{theorem}
%%%%%%%%%%%%%%%%%%%%%
\begin{proof}
 To establish the proof apply iteratively Lemma \ref{lemma2}.
\end{proof}
%%%%%%%%%%%%%%%%%%%%%%
\begin{remark}
 If the matrices $\Gamma_{o}$ and $\Gamma_{c}$ are chosen as zero and $\alpha=\beta$, then $S=\frac{1}{\alpha} Q$ and $T=\frac{1}{\alpha} P$. Hence, $Q\breve{P}=ST^{-1}$, and $\Lambda_{QP}=\Lambda_{ST}$. Accordingly, the error bound obtained via extended balancing coincides with the error bound obtained from the generalized balancing approach. Moreover, the reduced-order model obtained from both methods is the same. 
\end{remark}
%%%%%%%%%%%%%%%%%%%%%%
%%%%%%%%%%%%%%%%%%%%%%
Similar to the discrete-time results reported in \cite{SAND08} and \cite{SAND10}, the error bound \eqref{eb} is obtained by proposing a storage function 
and using dissipativity arguments, as in \cite{W02}. This procedure contrasts to the traditional analysis using transfer functions.
%%%%%%%%%%%%%%%%%%%%%%%%%
\section{Balancing of CTLTI PH systems}
\label{sec:bfphs}
%%%%%%%%%%%%%%%%%%%%%%%

From now onwards, we focus on the study of PH systems. These systems have been proved to be suitable to capture physical phenomena in different domains while preserving conservation laws \cite{GEObook}, \cite{VAN}. 
In this framework, it is possible to represent large scale networks of complex physical systems and, at the same time, underscore the roles of the energy, the interconnection pattern, and the dissipation in the behavior of such systems. 
Moreover, the passivity property of these systems can be straightforwardly proved by selecting the Hamiltonian function as a
storage function. Thus, given the possible physical interpretation of the PH models and their geometrical properties, this framework is appealing from both points of view: the theoretical and the practical one.
Therefore, preserving the PH structure for the reduced order model is interesting for analysis purposes and might be useful to give an interpretation of the behavior of the reduced order system.
In this section, we aim to solve the model reduction problem of CTLTI PH systems while preserving the PH
structure for the reduced order system. Furthermore, in some cases, not only the PH structure is preserved, but more particular structures which permit to provide a physical interpretation of the reduced order model.
%%%%%%%%%%%%%%%%%%%%%%%%
\subsection{CTLTI PH systems}\label{sec:phsys}
%%%%%%%%%%%%%%%%%%%%%%%
The representation of a CTLTI PH system is given by
\begin{equation}
\varSigma_{H}:\left\lbrace\begin{array}{rcl}
\dot{x}&=& (J-R)H x+Bu \\[.07cm] y&=& B^{\top}H x \\[.07cm] \mathcal{H}(x)&=&\frac{1}{2}x^{\top}H x        
       \end{array}\right. \label{sys} 
\end{equation}
where $x\in\rea^{n}$ is the state vector, $u,y \in \rea^{m}$ are the input and output vectors, respectively, $\mathcal{H}(x)$ represents the Hamiltonian of the system, with $H = H^{\top}>0$; and $R=R^{\top}\geq 0$, $J=-J^{\top}$ represent the dissipation and the interconnection matrix, respectively. 
In order to simplify notation, we define
$F:=J-R$.\\[0.2cm]
The objective of this work is twofold: on the one hand, we aim to balance system \eqref{sys} and obtain a lower order model. On the other hand, we want the reduced model to have a PH structure
because of the interpretation and the interconnection properties of this kind of systems.
Towards this end, we assume that system \eqref{sys} is asymptotically stable and we look for an invertible linear transformation $W$ that balances the system. Such transformation is given by $W=W_{g}$ in the generalized case, while in the extended case we have $W=W_{e}$.
Then, we write the dynamics of the balanced system as follows
\begin{equation*}\bar{\Lambda}_{H}:\left\lbrace
 \begin{array}{rcl}
   \dot{\bar x}
&=& \bar F \bar H\bar x + \bar B u
 \\\bar y&=&
       \bar B^{\top}\bar H \bar x,
 \end{array}\right. 
\end{equation*} 
where
\begin{equation*}
\begin{array}{rcl}
 \bar F:=W^{-1} F W^{-\top}, & \bar H:=W^{\top} H W, & \bar B:=W^{-1}B. 
\end{array}      
\end{equation*} 
Hence, if we split $\bar x$ as in \eqref{split}, the balanced system can be expressed as
\begin{equation}\bar{\varSigma}_{H}:\left\lbrace
 \begin{array}{rcl}
  \begin{bmatrix}
   \dot{\bar x}_{1} \\ \dot{\bar x}_{2}
  \end{bmatrix}
&=& \begin{bmatrix}
     \bar F_{11} & \bar F_{12} \\ \bar F_{21} & \bar F_{22}
    \end{bmatrix}
\begin{bmatrix}
\bar H_{11} & \bar H_{12} \\ \bar H_{12}^{\top} & \bar H_{22}
\end{bmatrix}
\begin{bmatrix}
\bar x_{1} \\ \bar x_{2}
\end{bmatrix}+
\begin{bmatrix}
\bar B_{1} \\ \bar B_{2}
\end{bmatrix}
u \\[.3cm]\bar y&=&\begin{bmatrix}
       \bar B_{1}^{\top} & \bar B_{2}^{\top}
       \end{bmatrix}
\begin{bmatrix}
\bar H_{11} & \bar H_{12} \\ \bar H_{12}^{\top} & \bar H_{22}
\end{bmatrix}\begin{bmatrix}
\bar x_{1} \\ \bar x_{2}
\end{bmatrix},
 \end{array}\right.\label{bsys} 
\end{equation} 
with
\begin{equation*}
 \begin{array}{lll}
 \bar F_{11},\bar H_{11}\in\rea^{k\times k}, & \bar F_{22}, \bar H_{22}\in\rea ^{\ell\times \ell}, &
 \bar F_{12}, \bar H_{12}\in\rea^{k\times \ell}, \\[0.1cm]  \bar B_{1}\in \rea^{k \times m}, &
 \bar F_{21}\in\rea^{\ell\times k}, & \bar B_{2}\in \rea^{\ell \times m}.
 \end{array}
\end{equation*}

\textbf{Problem formulation for PH systems.}
Given the system \eqref{sys}, find an invertible linear transformation $W$, that performs the balancing of the system and at the same time satisfies 

\begin{equation}
 \bar H_{12}=\mathbf{0}_{k \times n-k}. \label{conLam}
\end{equation} 

Note that, if \eqref{conLam} holds, the truncation leads to the following reduced order system

\begin{equation}
 \widehat\varSigma_{H}:\left\lbrace\begin{array}{rcl}
\dot{\hat x}&=& \bar F_{11}\bar H_{11} \hat x+\bar B_{1}u \\[.07cm] \hat y&=&\bar B_{1}^{\top}\bar H_{11}\hat x \\[.07cm] \hat{\mathcal{H}}(\hat x)&=&\frac{1}{2}\hat x^{\top}\bar H_{11} \hat x,        
       \end{array}\right. \label{phsysr} 
\end{equation}

which is another CTLTI PH system, with $\hat x = \bar x_{1}$. Therefore, it follows that one solution to the problem of model reduction with PH structure preservation takes place when the Hamiltonian matrix of the balanced system, $\bar{H}$, is diagonal. 
In such case, our problem is reduced to the simultaneous diagonalization of three matrices, namely, $(Q,P,H)$ or $(S,T,H)$.\\[.2cm]
%%%%%%%%%%%%%%%%
\begin{remark}
 The complete diagonalization of $H$ is not necessary. In fact, a block diagonalization that ensures \eqref{conLam} is enough to preserve the PH structure. 
 Nevertheless, if $H$ is not a diagonal matrix, then it is necessary to know the dimension of the part of the state to be truncated.\\
\end{remark}
%%%%%%%%%%%%%%%
The subsequent sections of this paper are devoted to the identification of a transformation $W$ that balances the system and ensures that \eqref{conLam} is satisfied. 
%%%%%%%%%%%%%%%%%%%%%%%%%%%%%%%%%
%%%%%%%%%%%%%%%%%%%%%%%%%%%%%%%%%
%%%%%%%%%%%%%%%%%%%%%%%%%%%%%%%%%
\subsection{Generalized balancing of CTLTI PH systems}\label{sec:gbph}
%%%%%%%%%%%%%%%%%%%%%%%%%%%%%%%%%
%%%%%%%%%%%%%%%%%%%%%%%%%%%%%%%%%
In this subsection, we study the generalized balancing method for CTLTI PH systems which is the starting point of extended balancing of CTLTI PH studied in Section \ref{sec:ebph}. 
Below, we provide sufficient conditions to ensure the existence of a transformation $W_{g}$ that complies with the requirements established in Section \ref{sec:phsys}. To this end, we revisit
the following theorem which establishes necessary and sufficient conditions for the existence of a transformation that diagonalizes simultaneously three matrices when
at least one of them has definite sign.\\
\begin{theorem}[\cite{NOV}]\label{th:diag}
 Let $L, M, N$ be symmetric matrices. In the case of at least one fixed-sign quadratic form (e.g., $M$ positive definite), 
 the condition 
  \begin{equation}
  LM^{-1}N=NM^{-1}L \label{3diag}
 \end{equation}
 is necessary and sufficient
 for the existence of a linear invertible congruent transformation $W$ that diagonalizes simultaneously $L, M$ and $N$.
\end{theorem}
For the proof and further details about Theorem \ref{th:diag}, we refer the reader to \cite{NOV} and \cite{CAU}. For a thorough exposition on simultaneously diagonalizable matrices, we refer
the reader to \cite{HORN}, Chapter 4.\\[0.2cm]
In generalized balancing of CTLTI PH systems, the condition \eqref{3diag} takes the form
\begin{equation}
H \breve P^{-1} Q = Q \breve P^{-1}H.\label{commute}
\end{equation} 
Accordingly, we look for $Q$ and $\breve P$ verifying \eqref{go} and \eqref{gc}, respectively, such that \eqref{commute} holds. A trivial solution to this problem takes place when
$Q$ or $\breve P$ coincides with the (scaled) Hamiltonian matrix $H$ or its inverse. This idea has been studied in \cite{F08} and \cite{KS18}, among other works; and for the sake of completeness, the proposition below identifies a class of CTLTI PH systems for which the (scaled) Hamiltonian matrix, or its inverse, solves the inequalities \eqref{go} and \eqref{gc}.\\
%%%%%%%%%%%%%%%%%%%%%%%%%%%%%%%%%
%%%%%%%%%%%%%%%%%%%%%%%%%%%%%%%%%
%%%%%%%%%%%%%%%%%%%%%%%%%%%%%%%%%
\begin{proposition}\label{proplam}\em
Consider $\delta\in \rea_{>0}$. Assume that the system \eqref{sys} is asymptotically stable. If the following condition holds
\begin{equation}
 2\delta R-BB^{\top}\geq0. \label{cond} 
\end{equation} 
 Then $Q=\delta H$ solves \eqref{go} and $\breve P=\delta H^{-1}$ is a solution to \eqref{gc}.
\end{proposition}
%%%%%%%%%%%%%%%%%%%%%%%%%%%%%%%%%
\begin{proof}
 To establish the proof note that for CTLTI PH systems \eqref{go} and \eqref{gc} take the form

 \begin{eqnarray}
 QFH+H F^{\top}Q+H BB^{\top}H &\leq& 0 \label{gog} \\
 FH \breve P+\breve PH F^{\top}+BB^{\top} &\leq& 0, \label{gcg} 
\end{eqnarray}
respectively. Hence, substituting $Q=\delta H$ in \eqref{gog}, we obtain

\begin{equation*}
 \begin{array}{rcl}
  0&\geq&\delta H FH+\delta H F^{\top}H+H BB^{\top}H \\
  &=&H(BB^{\top}-2 \delta R)H\\
  \Longleftrightarrow 0&\leq& 2\delta R-BB^{\top}.
 \end{array}
\end{equation*} 
 
 On the other hand, replacing $\breve P=\delta H^{-1}$ in \eqref{gcg}, we have
 
 \begin{equation*}
  \begin{array}{rcl}
   0&\geq&\delta F+\delta F^{\top}+BB^{\top}\\
   &=& -2\delta R+BB^{\top} \\ 
   \Longleftrightarrow 0&\leq& 2\delta R-BB^{\top}.
  \end{array}
 \end{equation*} 
\end{proof}
%%%%%%%%%%%%%%%%%%%%%%%%%%%%%%%%%
%%%%%%%%%%%%%%%%%%%%%%%%%%%%%%%%%
%%%%%%%%%%%%%%%%%%%%%%%%%%%%%%%%%
Condition \eqref{cond} is satisfied by systems that have dissipation in all the input channels, e.g., fully damped mechanical systems. Nonetheless,
$R$ and $B$ are system parameters, thus, 
it might happen that condition \eqref{cond} is not satisfied by the system \eqref{sys}. In order to overcome this issue, 
below we state two propositions to identify generalized Gramians such that the triplet $(Q, \breve{P}, H)$ verifies \eqref{commute} and solves the Lyapunov inequalities \eqref{gog} and \eqref{gcg}. These propositions represent the main result of this paper in terms of generalized balancing with PH structure preservation\\
%%%%%%%%%%%%%%%%%%%%%%%%%%%%%%%%%
%%%%%%%%%%%%%%%%%%%%%%%%%%%%%%%%%
%%%%%%%%%%%%%%%%%%%%%%%%%%%%%%%%%
\begin{proposition}\label{gobs}
 Let $\breve P$ be a solution to \eqref{gcg}. Consider a full rank matrix $\phi_{P}\in \rea^{n\times n}$ verifying the following
 
 \begin{equation*}
\begin{array}{rcl}
    \breve P&=&\phi_{P}^{\top}\phi_{P} \\
   \phi_{P}H\phi_{P}^{\top}&=&U_{H P}\Lambda_{H P}U_{H P}^{\top},
\end{array} 
 \end{equation*} 
 where $U_{H P}$ is an orthogonal matrix, and $\Lambda_{H P}$ is a diagonal matrix whose entries are the singular values of $\phi_{P}H\phi_{P}^{\top}$, see the notation at the end of Section \ref{sec:int}.
 Define the matrices 
 \begin{equation}
  \begin{array}{rcl}
   \mathcal{F}_{c}&:=&U_{H P}^{\top}\phi_{P}^{-\top}F\phi_{P}^{-1}U_{H P} \\[.1cm] \mathcal{B}_{c}&:=&U_{H P}^{\top}\phi_{P}^{-\top}B.
  \end{array}\label{exgo}
 \end{equation} 
 Assume that
 \begin{equation}
  -\Lambda_{QP}^{2}\Lambda_{H P}^{-1}\mathcal{F}_{c}-\mathcal{F}_{c}^{\top}\Lambda_{H P}^{-1}\Lambda_{QP}^{2}-\mathcal{B}_{c}\mathcal{B}_{c}^{\top}\geq 0 \label{cexgo} 
 \end{equation} 
 holds for a diagonal matrix $\Lambda_{QP}$. Hence, \eqref{gog} is solved by
 \begin{equation}
  Q=\phi_{P}^{-1}U_{H P}\Lambda_{QP}^{2}U_{H P}^{\top}\phi_{P}^{-\top}. \label{Q}
 \end{equation} 
 Moreover, the transformation
 \begin{equation}
  W_{gc}=\phi_{P}^{\top}U_{H P}\Lambda_{QP}^{-\frac{1}{2}} \label{Wc} 
 \end{equation} 
 balances the system and diagonalizes $H$.
\end{proposition}
%%%%%%%%%%
\begin{proof}
 To establish the proof we define
 \begin{equation*}
  \mathcal{X}_{o}:=-\Lambda_{QP}^{2}\Lambda_{H P}^{-1}\mathcal{F}_{c}-\mathcal{F}_{c}^{\top}\Lambda_{H P}^{-1}\Lambda_{QP}^{2}-\mathcal{B}_{c}\mathcal{B}_{c}^{\top}.
 \end{equation*} 
 Note that, if \eqref{cexgo} holds, we have the following chain of implications
 \begin{equation*}
\begin{array}{rcl}
   \mathcal{X}_{o}&\geq&0\\[.1cm]
  \Longleftrightarrow  \phi_{P}^{-1}U_{H P}\Lambda_{H P}\mathcal{X}_{o} \Lambda_{H P}U_{H P}^{\top}\phi_{P}^{-\top}  &\geq&0\\[.1cm]
  \Longleftrightarrow -QFH-H F^{\top}Q-H BB^{\top}H&\geq&0\\\Longleftrightarrow QFH+H F^{\top}Q+H BB^{\top}H&\leq&0
\end{array}
 \end{equation*} 
 where we used \eqref{exgo} and \eqref{Q}. Moreover,
 
 \begin{equation*}
  \begin{array}{rcl}
   W_{gc}^{\top}QW_{gc}&=& \Lambda_{QP}\\[.1cm]
   W_{gc}^{-1}\breve PW_{gc}^{-\top}&=&\Lambda_{QP}\\[.1cm]
   W_{gc}^{\top}H W_{gc}&=&\Lambda_{QP}^{-1}\Lambda_{H P}.
  \end{array}
 \end{equation*} 
 This completes the proof.
\end{proof}
%%%%%%%%%%%%%%%%%%%%%%%%%%%%%%%%%
%%%%%%%%%%%%%%%%%%%%%%%%%%%%%%%%%
%%%%%%%%%%%%%%%%%%%%%%%%%%%%%%%%%
The following proposition is the dual version of Proposition \ref{gobs} and relaxes condition \eqref{cond}, in this case, for a given generalized observability Gramian $Q$.\\[.2cm]
%%%%%%%%%%%%%%%%%%%%%%%%%%%%%%%%%
%%%%%%%%%%%%%%%%%%%%%%%%%%%%%%%%%
%%%%%%%%%%%%%%%%%%%%%%%%%%%%%%%%%
\begin{proposition}\label{gctr}
 Let $Q$ be a solution to \eqref{gog}. Consider a full rank matrix $\phi_{Q}\in \rea^{n \times n}$ verifying the following
 
 \begin{equation*}
 \begin{array}{rcl}
    Q&=&\phi_{Q}^{\top}\phi_{Q} \\
 \phi_{Q}^{-\top}H\phi_{Q}^{-1}&=&U_{H Q}\Lambda_{H Q}U_{H Q}^{\top}.
 \end{array}
 \end{equation*} 
 
 Define the matrices 
 
 \begin{equation}
  \begin{array}{rcl}
   \mathcal{F}_{o}&:=&U_{H Q}^{\top}\phi_{Q}F\phi_{Q}^{\top}U_{H Q} \\[.1cm] \mathcal{B}_{o}&:=&U_{H Q}^{\top}\phi_{Q}B.
  \end{array}\label{exgc}
 \end{equation} 
 
 Assume that
 
 \begin{equation}
  -\mathcal{F}_{o}\Lambda_{H Q}\Lambda_{QP}^{2}-\Lambda_{QP}^{2}\Lambda_{H Q}\mathcal{F}_{o}^{\top}-\mathcal{B}_{o}\mathcal{B}_{o}^{\top}\geq 0 \label{cexgc} 
 \end{equation} 
 holds for a diagonal matrix $\Lambda_{QP}$. Hence, \eqref{gcg} is solved by
 
 \begin{equation}
  \breve P=\phi_{Q}^{-1}U_{H Q}\Lambda_{QP}^{2}U_{H Q}^{\top}\phi_{Q}^{-\top}. \label{P}
 \end{equation} 
 
 Moreover, the transformation
 
 \begin{equation}
  W_{go}=\phi_{Q}^{-1}U_{H Q}\Lambda_{QP}^{\frac{1}{2}} \label{Wo} 
 \end{equation} 
 
 balances the system and diagonalizes $H$.
\end{proposition}
%%%%%%%%%%%%%%%%%%%%%%%%
\begin{proof}
 Define 
 \begin{equation}
  \mathcal{X}_{c}:=-\mathcal{F}_{o}\Lambda_{H Q}\Lambda_{QP}^{2}-\Lambda_{QP}^{2}\Lambda_{H Q}\mathcal{F}_{o}^{\top}-\mathcal{B}_{o}\mathcal{B}_{o}^{\top}.
 \end{equation} 
 Therefore, if \eqref{cexgc} is satisfied, we have
 \begin{equation*}
\begin{array}{rcl}
   \mathcal{X}_{c}&\geq&0\\[.1cm]
  \Longleftrightarrow  \phi_{Q}^{-1}U_{H Q}\mathcal{X}_{c} U_{H Q}^{\top}\phi_{Q}^{-\top}  &\geq&0\\[.1cm]
  \Longleftrightarrow -FH\breve P-\breve PH F^{\top}- BB^{\top}&\geq&0\\\Longleftrightarrow FH\breve P+\breve PH F^{\top}+ BB^{\top}&\leq&0,
\end{array}
 \end{equation*}
 where we used \eqref{exgc} and \eqref{P}. To complete the proof, note that
 \begin{equation*}
  \begin{array}{rcl}
   W_{go}^{\top}QW_{go}&=&  \Lambda_{QP}\\[0.1cm]
   W_{go}^{-1}\breve PW_{go}^{-\top}&=&\Lambda_{QP}\\[0.1cm]
   W_{go}^{\top}H W_{go}&=&  \Lambda_{H Q}\Lambda_{QP}.
  \end{array}
 \end{equation*}  
\end{proof}
%%%%%%%%%%%%%%%%%%%%%%%%%%%%%%%%%
%%%%%%%%%%%%%%%%%%%%%%%%%%%%%%%%%
%%%%%%%%%%%%%%%%%%%%%%%%%%%%%%%%%
% \begin{remark}
%  The generalized Gramians obtained in Propositions \ref{gobs} and \ref{gctr} satisfy
% \begin{equation}
%  H P Q = Q PH\label{conmute}
% \end{equation}  
% Therefore, the aforementioned propositions provide sufficient conditions to tackle down two problems, that is, to find solutions $Q$ and $\breve{P}$ to \eqref{go} and \eqref{gc}, respectively, while
% ensuring that those solutions satisfy \eqref{conmute}.
% \end{remark}
In Propositions \ref{gobs} and \ref{gctr}, the condition \eqref{cond} is relaxed by imposing a particular structure to the generalized observability and controllability Gramians, respectively.
Such structure depends on the Hamiltonian matrix, however, it is less restrictive than \eqref{cond}. Indeed, if this latter condition is satisfied, then \eqref{cexgo} and \eqref{cexgc} hold.\\[0.2cm]
Using the results presented in this section, below we study extended balancing of CTLTI PH systems.
As was mentioned in Section \ref{sec:eg}, the use of extended Gramians can be advantageous for different purposes, for instance, to obtain a lower error bound
or to impose a more particular structure to the reduced order model.
%%%%%%%%%%%%%%%%%%%%%%%%%%%%%%%%%
%%%%%%%%%%%%%%%%%%%%%%%%%%%%%%%%%
%%%%%%%%%%%%%%%%%%%%%%%%%%%%%%%%%
\subsection{Extended balancing of CTLTI PH systems}
\label{sec:ebph}
%%%%%%%%%%%%%%%%%%%%%%%%%%%%%%%%%
%%%%%%%%%%%%%%%%%%%%%%%%%%%%%%%%%

Similar to the generalized balancing case, in this section we provide sufficient conditions for the existence of a linear transformation $W_{e}$ that balances the system and diagonalizes the Hamiltonian matrix. Towards this end, below we introduce 
two propositions that provide a suitable transformation $W_{e}$. Such propositions constitute the main result of this work regarding extended balancing with PH structure preservation.\\
%%%%%%%%%%%%%%%%%%%%%%%%%%%%%%%%%
%%%%%%%%%%%%%%%%%%%%%%%%%%%%%%%%%
%%%%%%%%%%%%%%%%%%%%%%%%%%%%%%%%%
\begin{proposition}\label{prop:ex1} 
 Let $\breve P$ be a solution to \eqref{gcg} such that $X_{c}>0$. Select $\beta$ and $\Gamma_c$ such that \eqref{posT} holds and $T$, defined in \eqref{Tsel}, solves LMI \eqref{LMIcontr}.
 Consider a full rank matrix $\phi_{T}\in \rea^{n\times n}$ verifying the following
 
 \begin{equation*}
  \begin{array}{rcl}
   T^{-1}&=&\phi_{T}^{\top}\phi_{T}\\
   \phi_{T}H\phi_{T}^{\top}&=&U_{H T}\Lambda_{H T}U_{H T}^{\top}.
  \end{array}
 \end{equation*} 
 
 Define the matrices 
 \begin{equation}
  \begin{array}{rcl}
   \mathcal{F}_{ec}&:=&U_{H T}^{\top}\phi_{T}^{-\top}F\phi_{T}^{-1}U_{H T} \\[.1cm] \mathcal{B}_{ec}&:=&U_{H T}^{\top}\phi_{T}^{-\top}B.
  \end{array}\label{Fec}
 \end{equation} 
 Assume that
 \begin{equation}
  -\Lambda_{QT}^{2}\Lambda_{H T}^{-1}\mathcal{F}_{ec}-\mathcal{F}_{ec}^{\top}\Lambda_{H T}^{-1}\Lambda_{QT}^{2}-\mathcal{B}_{ec}\mathcal{B}_{ec}^{\top}> 0 \label{cexeo}
 \end{equation} 
 holds for a diagonal matrix $\Lambda_{QT}$. Then, \eqref{gog} is solved by
 \begin{equation}
  Q=\phi_{T}^{-1}U_{H T}\Lambda_{QT}^{2}U_{H T}^{\top}\phi_{T}^{-\top}. \label{Qe}
 \end{equation} 
 Select $\alpha$ such that the matrix
 \begin{equation}
 S=\frac{1}{\alpha}Q \label{Se}
 \end{equation} 
 solves LMI \eqref{LMIobs}. Then, the invertible transformation 
 \begin{equation}
  W_{ec}=\sqrt[4]{\alpha}\phi_{T}^{\top}U_{H T}\Lambda_{QT}^{-\frac{1}{2}} \label{Wec}
 \end{equation} 
 balances the system and diagonalizes $H$.   
\end{proposition}
%%%%%%%%%%%%%%%%%%%%%%%%%%%%%%%%%
\begin{proof}
 Define
 \begin{equation*}
  \mathcal{X}_{eo}:=-\Lambda_{QT}^{2}\Lambda_{H T}^{-1}\mathcal{F}_{ec}-\mathcal{F}_{ec}^{\top}\Lambda_{H T}^{-1}\Lambda_{QT}^{2}-\mathcal{B}_{ec}\mathcal{B}_{ec}^{\top}.
 \end{equation*} 
 Then, the inequality \eqref{cexeo} is satisfied if and only if
 \begin{equation}
\begin{array}{rcl}
\mathcal{X}_{eo}&>&0\\[0.1cm]
   \phi_{T}^{-1}U_{H T}\Lambda_{H T}\mathcal{X}_{eo} \Lambda_{H T}U_{H T}^{\top}\phi_{T}^{-\top}  &>&0\\[.1cm]
  \Longleftrightarrow X_o&>&0 \\[0.1cm]
  \Longleftrightarrow QFH+H F^{\top}Q+H BB^{\top}H&<&0\\ \Longleftrightarrow QFH+H F^{\top}Q+H BB^{\top}H&\leq&0,
\end{array}
 \end{equation} 
 where we used 
 \begin{equation}
  A=FH. \label{A}
 \end{equation} 
 Fix $\Gamma_{o}=\mathbf{0}_{n\times n}$ in \eqref{symS}. Hence, for $\alpha$ large enough, the selection of $S$ given in \eqref{Se} solves the LMI \eqref{LMIobs}.

 To establish the last part of the proof define
 \begin{equation}
  \Lambda_{ST}:=\frac{1}{\sqrt{\alpha}}\Lambda_{QT},
 \end{equation} 
 note that
 \begin{equation}
  \begin{array}{rcl}
   W_{ec}^{-1}T^{-1}SW_{ec}&=& \Lambda_{ST}^{2} \\[.15cm]
   W_{ec}^{\top}H W_{ec}&=&\Lambda_{H T}\Lambda_{ST}^{-1}.
  \end{array}
 \end{equation} 
 \end{proof}
 %%%%%%%%%%%%%%%%%%%%%%%%%%%%%%%%%
%%%%%%%%%%%%%%%%%%%%%%%%%%%%%%%%%
%%%%%%%%%%%%%%%%%%%%%%%%%%%%%%%%%
The following proposition is the dual version of Proposition \ref{prop:ex1}.\\
%%%%%%%%%%%%%%%%%%%%%%%%%%%%%%%%%
%%%%%%%%%%%%%%%%%%%%%%%%%%%%%%%%%
%%%%%%%%%%%%%%%%%%%%%%%%%%%%%%%%%
\begin{proposition}\label{prop:ex2}
 Let $Q$ be a solution to \eqref{gog} such that $X_{o}>0$. Select $\alpha$ and $\Gamma_{o}$ such that \eqref{posQ} holds and $S$, defined in \eqref{symS}, solves LMI \eqref{LMIobs}.
 Consider a full rank matrix $\phi_{S}\in \rea^{n \times n}$ verifying the following
 
 \begin{equation*}
 \begin{array}{rcl}
  S&=&\phi_{S}^{\top}\phi_{S}\\
 \phi_{S}^{-\top}H\phi_{S}^{-1}&=&U_{H S}\Lambda_{H S}U_{H S}^{\top}.
 \end{array}
 \end{equation*} 
 
 Define the matrices 
 
 \begin{equation}
  \begin{array}{rcl}
   \mathcal{F}_{eo}&:=&U_{H S}^{\top}\phi_{S}F\phi_{S}^{\top}U_{H S} \\[.1cm] \mathcal{B}_{eo}&:=&U_{H S}^{\top}\phi_{S}B.
  \end{array}\label{exec}
 \end{equation} 
 
 Assume
 
 \begin{equation}
  -\mathcal{F}_{eo}\Lambda_{H S}\Lambda_{SP}^{2}-\Lambda_{SP}^{2}\Lambda_{H S}\mathcal{F}_{eo}^{\top}-\mathcal{B}_{eo}\mathcal{B}_{eo}^{\top}> 0 \label{cexec} 
 \end{equation} 
 holds for a diagonal matrix $\Lambda_{SP}$. Thus, \eqref{gcg} is solved by
 
 \begin{equation}
  \breve P=\phi_{S}^{-1}U_{H S}\Lambda_{SP}^{2}U_{H S}^{\top}\phi_{S}^{-\top}. \label{Pe}
 \end{equation} 
 
Select $\beta$ such that the matrix
 \begin{equation}
   T^{-1}=\beta \breve P \label{Te}
 \end{equation} 
 solves LMI \eqref{LMIcontr}. Then, 
 
  \begin{equation}
  W_{eo}=\sqrt[4]{\beta}\phi_{S}^{-1}U_{H S}\Lambda_{SP}^{\frac{1}{2}} \label{Weo}
 \end{equation} 
 
 balances the system and diagonalizes $H$.   
\end{proposition}
%%%%%%%%%%%%%%%%%%%%%%%%%%%%%%%%%
\begin{proof}
 Define
 \begin{equation*}
  \mathcal{X}_{eo}:=-\mathcal{F}_{eo}\Lambda_{H S}\Lambda_{SP}^{2}-\Lambda_{SP}^{2}\Lambda_{H S}\mathcal{F}_{eo}^{\top}-\mathcal{B}_{eo}\mathcal{B}_{eo}^{\top}.
 \end{equation*} 
 Hence, if \eqref{cexec} holds, we have the following chain of implications
 \begin{equation}
  \begin{array}{rcl}
   \mathcal{X}_{eo}&>&0 \\[0.1cm]
   \Longleftrightarrow\phi_{S}^{-1}U_{H S}\mathcal{X}_{co}U_{H S}^{\top}\phi_{S}^{-\top}&>&0\\[0.1cm]
   \Longleftrightarrow -FH \breve P-\breve PH F^{\top}-BB^{\top}&>&0.
  \end{array}
 \end{equation} 
 Moreover,
 \begin{equation}
 \begin{array}{l}
  -FH \breve P-\breve PH F^{\top}-BB^{\top}>0 \hfill \\[0.1cm] \hfill  \; \; \; \; \; \; \; \; \; \; \; \; \; \; \; \; \; \;   \Longrightarrow \left\lbrace \begin{array}{rcl}
  FH \breve P+\breve PH F^{\top}+BB^{\top}\leq 0 \\[0.1cm]                                                                                          
  X_{c}>0,                                                                                         \end{array}\right.
  \end{array}
 \end{equation}
 where we used \eqref{A}.
 Fix $\Gamma_{c}=\mathbf{0}_{n\times n}$ in \eqref{Tsel}. Accordingly, for $\beta$ large enough, the selection of $T$ given in \eqref{Te} solves the LMI \eqref{LMIcontr}.

 To establish the last part of the proof define
 \begin{equation}
  \Lambda_{ST}:=\sqrt{\beta}\Lambda_{SP}.
 \end{equation} 
 Note that
 \begin{equation}
  \begin{array}{rcl}
   W_{eo}^{-1}T^{-1}SW_{eo}&=& \Lambda_{ST}^{2} \\[.15cm]
   W_{eo}^{\top}H W_{eo}&=&\Lambda_{H S}\Lambda_{ST}.
  \end{array}
 \end{equation} 
\end{proof}
%%%%%%%%%%%%%%%%%%%%%%%%%%%%%%%%%
%%%%%%%%%%%%%%%%%%%%%%%%%%%%%%%%%
%%%%%%%%%%%%%%%%%%%%%%%%%%%%%%%%%
We remark that $\Gamma_{o}$ and $\Gamma_{c}$ are degrees of freedom in the selection of $S$ and $T$, respectively. 
These matrices can be selected in order to improve the error bound or preserve more particular structures as is illustrated in Section \ref{sec:ex}.
\section{Examples}
\label{sec:ex}
%%%%%%%%%%%%%%%%%%%%%%%%

In this section we present two examples to illustrate the applicability of the results reported in previous sections. Both examples represent physical systems, where, the first one is a mass-spring-damper 
mechanical system. While, the second example represents an RLC circuits network. 

\subsection{Mechanical system}\label{sec:mecex}

Consider five mass-spring-damper systems interconnected in series as shown in Fig. \ref{fig:mechanical}. The dynamics that describe this network of mechanical systems are given by

\begin{equation}
 \begin{array}{rcl}
 \begin{bmatrix}
  \dot{q} \\ \dot{p}
 \end{bmatrix}
&=&\underbrace{\begin{bmatrix}
 \mathbf{0}_{5\times 5} & I_{5} \\ -I_{5} & -R_{2}  
   \end{bmatrix}}_{F}\underbrace{\begin{bmatrix}
                 K & \mathbf{0}_{5\times 5} \\\mathbf{0}_{5\times 5} & M^{-1}
                \end{bmatrix}}_{H}\begin{bmatrix}
                                         q \\ p
                                        \end{bmatrix}
+ \underbrace{\begin{bmatrix}
\mathbf{0}_{5} \\ G                 \end{bmatrix}
}_{B}u \\[1cm] G&=&\begin{bmatrix}
                   1 \\ \mathbf{0}_{4}
                   \end{bmatrix}, \; \; M=\diag\{m_{1},m_{2},m_{3},m_{4},m_{5}\}
\\[.5cm]                                                                                        K&=&\begin{bmatrix}
    k_1 & -k_1 & 0 & 0 & 0\\ -k_1 & k_1+k_2 & -k_2 & 0 & 0\\ 0 & -k_2 & k_2+k_3 & -k_3 & 0\\  0 & 0 & -k_3 & k_3+k_4 & -k_4\\ 0 & 0 & 0 & -k_4 & k_4+k_5                                                              
 \end{bmatrix},\\ [1cm] 
%  M&=&\diag\{m_{1},m_{2},m_{3},m_{4},m_{5}\}\\[.5cm]
R_2&=&\begin{bmatrix}
       0 & 0 & 0 & 0 & 0\\ 0 & b_2 & -b_2 & 0 & 0\\ 0 & -b_2 & b_2+b_3 & -b_3 & 0\\
    0 & 0 & -b_3 & b_3+b_4 & -b_4\\ 0 & 0 & 0 & -b_4 & b_4
      \end{bmatrix},
 \end{array}\label{mec}
\end{equation} 
where $q,p\in\rea^{5}$, which is in PH form.

\begin{figure*}
 \centering
 \captionsetup{justification=centering}
 \includegraphics[width=\textwidth]{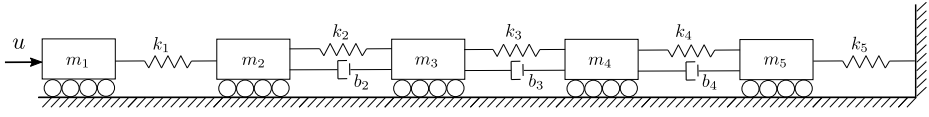}
 \vspace{.2cm}
  \caption{Mass-spring-damper network}
 \label{fig:mechanical}
\end{figure*}

The objective is to reduce the order of the model and ensure that the PH structure is preserved.  
Note that, independently of $\delta>0$, this system does not satisfy condition \eqref{cond}, and thus the Hamiltonian matrix cannot be proposed as a generalized Gramian.
At this point, we remark that this system is neither controllable nor observable, but Assumption \ref{A1} holds.

\begin{table}[h]
\caption{\\Parameters of the mechanical system}
 \begin{multicols}{3}
\begin{center}
\begin{tabular}{|l|l|} \hline
$b_{2}$ & $50 [kg/s]$\\ \hline
$b_{3}$ & $20 [kg/s]$\\ \hline
$b_{4}$ & $5 [kg/s]$ \\ \hline
\end{tabular}
\end{center}
\begin{center}
\begin{tabular}{|l|l|} \hline
$m_{1}$ & $1.5 [kg]$\\ \hline
$m_{2}$ & $0.5 [kg]$\\ \hline
$m_{3}$ & $4 [kg]$\\ \hline
$m_{4}$ & $2 [kg]$ \\ \hline
$m_{5}$ & $1.25 [kg]$\\ \hline
\end{tabular}
\end{center}
\begin{center}
\begin{tabular}{|l|l|} \hline
$k_{1}$ & $4 [kg/s^{2}]$\\ \hline
$k_{2}$ & $7 [kg/s^{2}]$\\ \hline
$k_{3}$ & $2 [kg/s^{2}]$\\ \hline
$k_{4}$ & $5 [kg/s^{2}]$ \\ \hline
$k_{5}$ & $3 [kg/s^{2}]$\\ \hline
\end{tabular}
\end{center}
\end{multicols}
\label{T:mec}
\end{table}
Based on the results presented in Section \ref{sec:gbph}, we first adopt the generalized balanced truncation approach. To this end, we look for a solution $\breve P$ to the inequality \eqref{gcg} such that $X_{c}>0$. Hence, if the conditions established in Proposition \ref{gobs} are satisfied, then we propose $W=W_{gc}$.
In order to reduce the order of system \eqref{mec} via generalized balanced truncation, we proceed as follows:
\begin{itemize}
 \item We propose a positive definite matrix $\breve X_{c}\in\rea^{n}$ to write the inequality \eqref{gcg} as an equality, that is,
 \begin{equation}
 -BB^{\top}-\breve X_{c}=FH\breve P+\breve P H F^{\top}, \label{lyapmec}
\end{equation} 
\item We find $\breve P$ that solves \eqref{lyapmec}.
\item We look for a diagonal matrix that solves \eqref{cexgo}. If such diagonal matrix exists, then we propose $W=W_{gc}$, with $W_{gc}$ defined in \eqref{Wc}. Notice that the singular values of the system are contained in $\Lambda_{QP}$.
\item We truncate the system and we obtain the reduced order model.
\end{itemize}
For illustration purposes, we consider the values given in Table \ref{T:mec}, and we fix $\breve X_{c}=I_{10}\times 10^{-5}$. Hence, using Matlab, we find $\breve P$ that solves \eqref{lyapmec} which is given by \eqref{Ptmec}, see the Appendix. 
Moreover, we use Matlab to solve the inequality
\eqref{cexgo}, obtaining the solution 
\begin{equation}
 \begin{array}{rcl}
  \Lambda_{QP}&=&\diag \{4.374, 4.316, 2.755, 2.564,\\&& 1.188, 0.626, 0.482,  0.324, 0.155, 0.070\}. 
 \end{array}
\end{equation} 
Therefore, it follows from Proposition \ref{gobs} that the reduced order model preserves the PH structure.\\[0.2cm]
Now, for the sake of comparison we study the extended balancing case. Towards this end, we proceed as follows:

\begin{itemize}
 \item We consider the matrices $\breve X_{c}$ and $\breve P$ used during the generalized balancing procedure.
 \item We propose $\Gamma_{c}$ and $\beta$ such that $T$, given in \eqref{Tsel}, solves \eqref{LMIcontr}.
 \item We look for a diagonal matrix $\Lambda_{ST}$ that solves \eqref{cexeo}. If such matrix exists, then we select $S$ as in \eqref{Se},
with $Q$ given in \eqref{Qe}.
\item We truncate the system and we obtain the reduced order model.
 \end{itemize}

To illustrate the methodology, we replace $T=(\beta \breve P^{-1}+\Gamma_c)^{-1}$ in \eqref{LMIcontr}. Hence, using Matlab, we solve this equation for $\beta$ and a symmetric
matrix $\Gamma_{c}$. As a result, we obtain 
\begin{equation*}
 \beta=4.8021\times 10^{7}
\end{equation*}
and the matrix given in \eqref{Gamcmec}, see the Appendix. Then, we fix $\alpha=\beta$ and we look for a solution $\Lambda_{QT}$ to the inequality \eqref{cexeo}.
Such a diagonal matrix is obtained by Matlab's
LMI solver. Finally, we fix $\Lambda_{ST}=\frac{1}{\sqrt{\alpha}}\Lambda_{QT}$ to obtain
\begin{equation}
 \begin{array}{rcl}
  \Lambda_{ST}&=&\diag \{3.71, 3.666, 2.415, 2.218, 0.976,\\&& 0.543,  0.401, 0.245, 0.099, 0.041\}.
 \end{array}
\end{equation}
Thus, it follows from Proposition \ref{prop:ex1} that the reduced order model preserves the PH structure. In this example, we chose $\Gamma_{c}$ based on the value of the four smallest entries of $\Lambda_{ST}$. We tuned $\Gamma_{c}$ by trial and error for illustration purposes. However, to improve the results, this matrix can be computed by solving an optimization problem.\\[0.2cm]
In order to compare the error bounds of both balancing approaches, we truncate four states of the original system, that is $k=6$. Accordingly, for the generalized balancing case we get
\begin{equation}
 \lVert \varSigma - \varSigma_{r} \rVert_{\infty}\leq 2.06, \label{emecg}
\end{equation}
and for extended balancing we have 
\begin{equation}
 \lVert \varSigma - \varSigma_{r} \rVert_{\infty}\leq 1.57. \label{emece}
\end{equation} 
To compare the behavior of system \eqref{mec} and both reduced order systems, obtained via generalized balanced truncation and extended
balanced truncation, we perform simulations under initial conditions $\bar x=\mathbf{0}_{10}$ for the balanced system, $\hat x=\mathbf{0}_{6}$ for both reduced order systems, 
and the input signal depicted in Fig. \ref{fig:u}. Figure \ref{fig:yg_mec} shows the comparison between the outputs of the balanced system and the reduced order system obtained via 
generalized balanced truncation. Analogously in Fig. \ref{fig:ye_mec}, we plot the outputs of the balanced system and the system obtained through extended balanced truncation. 
From Figs. \ref{fig:yg_mec} and \ref{fig:ye_mec} we notice that the output of the balanced system, plotted in black, and the outputs of the reduced order systems, depicted in red,
significantly similar. Thus we conclude that, with both balanced truncation approaches, we preserve the PH structure of the original system, and the response of the reduced order systems 
to a given input is similar to the response of the original system.\\[0.2cm]
We remark that
while the PH structure is preserved for generalized and extended balanced truncation, the comparison between \eqref{emecg} and \eqref{emece} shows that the error bound
obtained from the latter balancing method is smaller. However, in the simulations we performed, for the real error, the smaller error bound did not make much difference. 
This situation can be observed in Fig. \ref{fig:error}, where we present in black the difference between the output of the balanced system and the output of the reduced order system
obtained through generalized balanced truncation, and in red we plot the difference between the output of the balanced system and the output of the reduced order system
obtained through extended balanced truncation.

\begin{figure}[h]
 \centering
 \includegraphics[scale=.34]{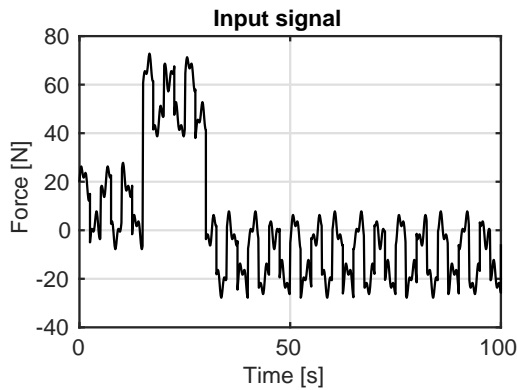}
 % umec.eps: 0x0 px, 300dpi, 0.00x0.00 cm, bb=
 \caption{Signal $u$.}
 \label{fig:u}
\end{figure}
\begin{figure}[h]
 \centering
 \includegraphics[scale=.34]{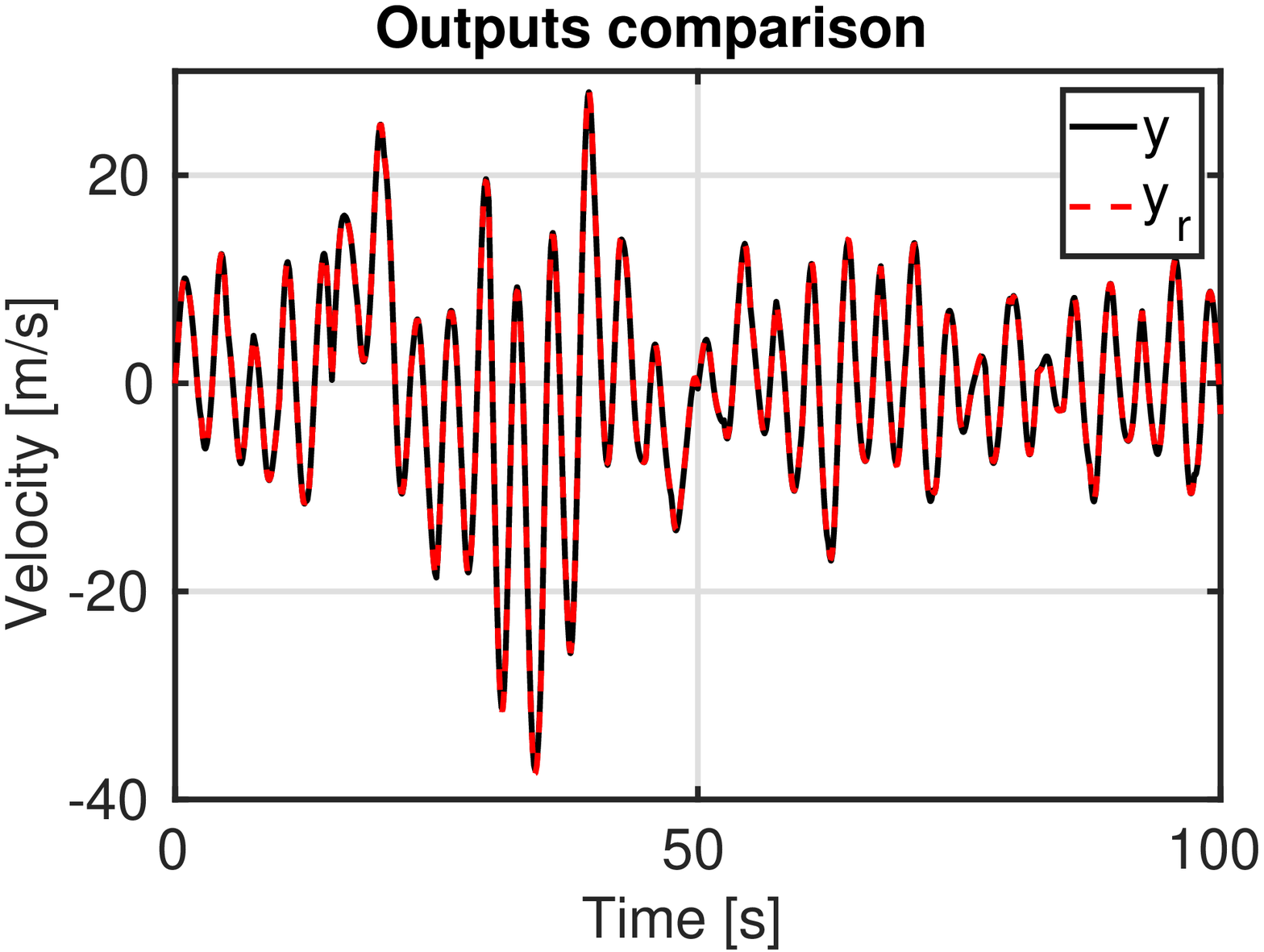}
 % ymec.eps: 0x0 px, 300dpi, 0.00x0.00 cm, bb=
 \caption{Outputs of the balanced system and the reduced order system obtained via generalized balanced truncation.}
 \label{fig:yg_mec}
\end{figure}
\begin{figure}[h]
 \centering
 \includegraphics[scale=.34]{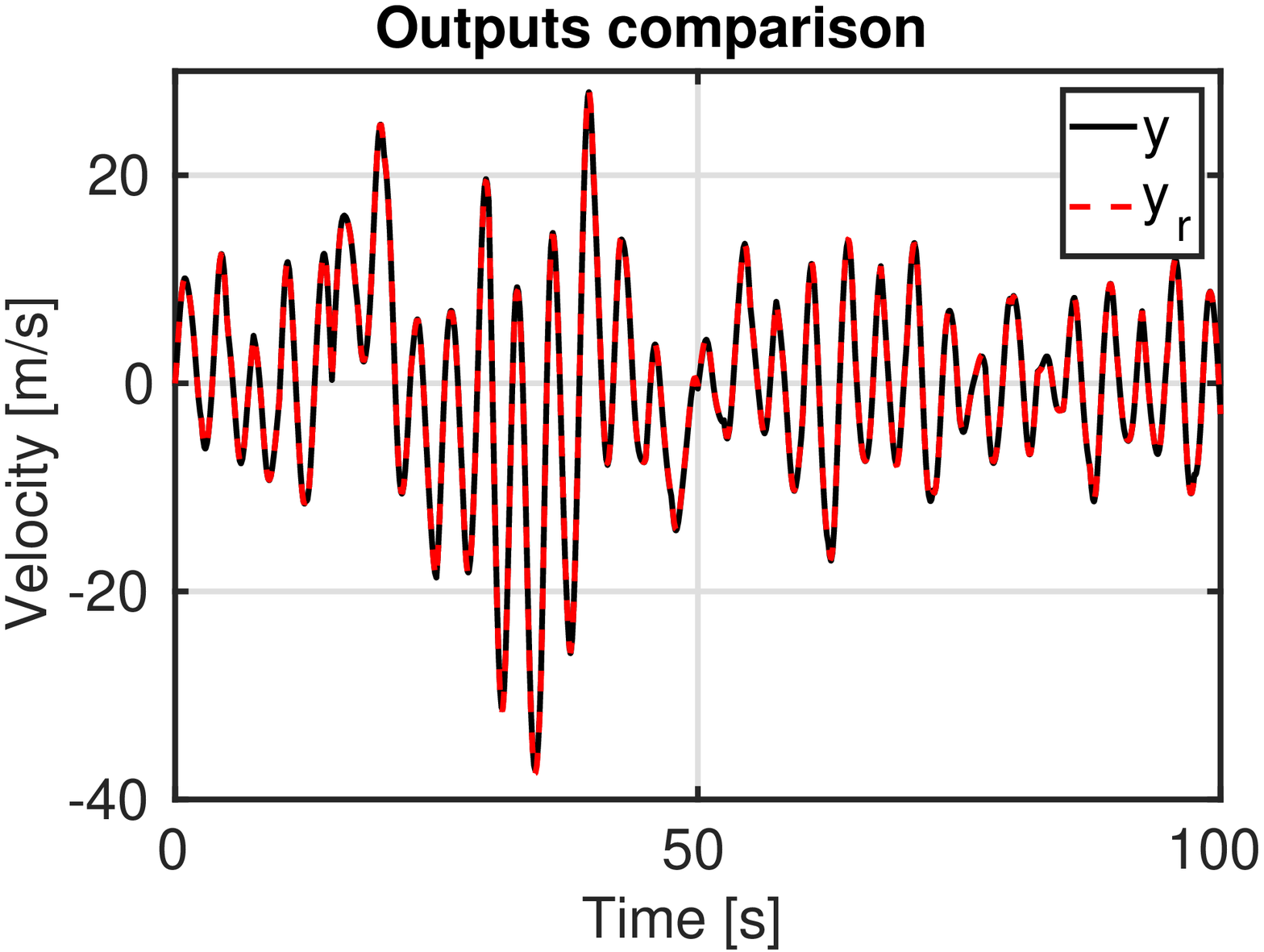}
 % ymec.eps: 0x0 px, 300dpi, 0.00x0.00 cm, bb=
 \caption{Outputs of the balanced system and the reduced order system obtained via generalized balanced truncation.}
 \label{fig:ye_mec}
\end{figure}
\begin{figure}
 \centering
 \includegraphics[scale=.34]{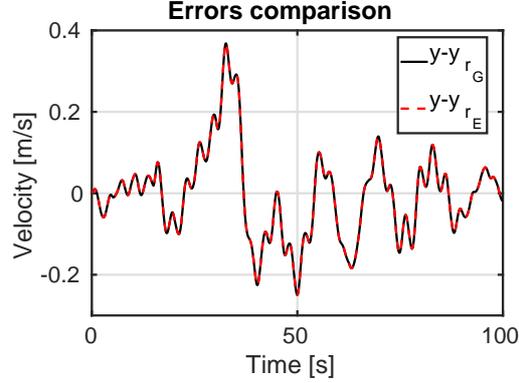}
 % umec.eps: 0x0 px, 300dpi, 0.00x0.00 cm, bb=
 \caption{Error between the outputs of the balanced system and the reduced order systems.}
 \label{fig:error}
\end{figure}

\subsection{RLC circuit}

\begin{figure*}
 \centering
 \captionsetup{justification=centering}
 \includegraphics[width=\textwidth]{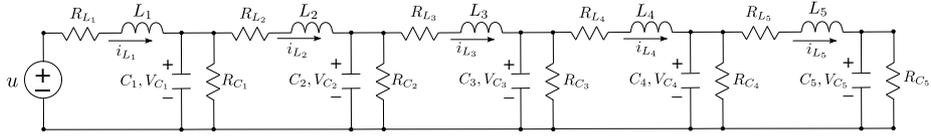}
 \vspace{.2cm}
  \caption{RLC network}
 \label{fig:RLC}
\end{figure*}

Consider the RLC network depicted in Fig. \ref{fig:RLC} which admits a PH representation of the form \eqref{sys} with

\begin{equation}
\begin{array}{rcl}
 J&=&\begin{bmatrix}
    \mathbf{0}_{5\times 5} & J_{1} \\ -J_{1}^{\top} & \mathbf{0}_{5\times 5}
   \end{bmatrix}, \; \;
   R= \begin{bmatrix}
         R_{C}^{-1} & \mathbf{0}_{5\times 5} \\ \mathbf{0}_{5\times 5} &R_{L}
        \end{bmatrix}, \; \;\\[.3cm]
        H&=&\diag\{ \frac{1}{C_{1}}, \frac{1}{C_{2}}, \frac{1}{C_{3}}, \frac{1}{C_{4}}, \frac{1}{C_{5}} ,
        \frac{1}{L_{1}}, \frac{1}{L_{2}}, \frac{1}{L_{3}}, \frac{1}{L_{4}}, \frac{1}{L_{5}}\}, \\[0.2cm]
  R_{C}&=& \diag\{ R_{C_{1}}, R_{C_{2}},R_{C_{3}}, R_{C_{4}}, R_{C_{5}} \},\\[.2cm]
 R_{L}&=& \diag\{ R_{L_{1}}, R_{L_{2}},R_{L_{3}}, R_{L_{4}}, R_{L_{5}} \},\\[.2cm]
 J_{1}&=& \begin{bmatrix}
           1 & -1 & 0 & 0 & 0 \\ 0 & 1 & -1 & 0 & 0 \\ 0 & 0 & 1 & -1 & 0 \\ 0 & 0 & 0 & 1 &  -1 \\ 0 & 0 & 0 & 0 & 1
          \end{bmatrix}, \; \; B= \begin{bmatrix}
        \mathbf{0}_{5} \\ 1 \\ \mathbf{0}4
       \end{bmatrix},
       \end{array}
\end{equation} 

where $x_{i}$ are the charges in the capacitors and $x_{5+i}$ denote the fluxes in the inductors, for $i=1,\cdots,5$.

\begin{table}[h]
\caption{\\Parameters of the RLC network}
\begin{multicols}{2}
\begin{center}
\begin{tabular}{|c|c|}\hline
$R_{C_{1}}$ & $270[\Omega]$\\\hline
$R_{C_{2}}$ & $1[k\Omega]$\\\hline
$R_{C_{3}}$ & $330[\Omega]$\\\hline
$R_{C_{4}}$ & $1.5[k\Omega]$\\\hline
$R_{C_{5}}$ & $220[\Omega]$\\\hline
$R_{L_{1}}$ & $4.7[\Omega]$\\\hline
$R_{L_{2}}$ & $3.9[\Omega]$\\\hline
$R_{L_{3}}$ & $2.2[\Omega]$\\\hline
$R_{L_{4}}$ & $2.74[\Omega]$\\\hline
$R_{L_{5}}$ & $3.92[\Omega]$\\\hline
\end{tabular}
\end{center}
\columnbreak
\begin{center}
\begin{tabular}{|c|c|}\hline
$C_{1}$ & $2.2[m F]$\\\hline
$C_{2}$ & $1[m F]$\\\hline
$C_{3}$ & $3.3[m F]$\\\hline
$C_{4}$ & $15[\mu F]$\\\hline
$C_{5}$ & $4.7[\mu F]$\\\hline
$L_{1}$ & $10[m H]$\\\hline
$L_{2}$ & $4.3[m H]$\\\hline
$L_{3}$ & $2.7[m H]$\\\hline
$L_{4}$ & $6.2[\mu H]$\\\hline
$L_{5}$ & $3[\mu H]$\\\hline
\end{tabular}
\end{center}
\end{multicols}
\label{T:rlc}
\end{table}
The objective is to reduce the order of the model and obtain a PH system that has a physical interpretation as an RLC circuit. Accordingly, we require that the reduced PH system has a diagonal damping matrix, and the interconnection matrix must be skew-symmetric and block anti-diagonal, which is more particular than the standard PH structure given in \eqref{sys}.
We stress the fact that the matrices $J,\ R$, and $H$ can be decomposed in block matrices whose dimension depend on the number of inductors and capacitors, in this case 5. Moreover,
$H$ is already diagonal. Thus, a block diagonal transformation\footnote{Where the dimension of the blocks is again related to the number of capacitors and inductors.} $W$ ensures that $\bar H$ remains
diagonal, and the block structure that determines the RLC architecture of the system is not affected.\\[0.2cm]
Note that the damping matrix $R$ has full rank. Hence, we can select
\begin{equation}
 \begin{array}{rl}
  Q=\delta_{o}H, & \breve P=\delta_{c}H^{-1}, \label{PRLC}
 \end{array}
\end{equation}
where $\delta_{o}$ and $\delta_{c}$ are positive constants such that \eqref{cond} holds. Therefore, both generalized Gramians are diagonal and the resulting transformation $W_{g}$ will not modify the structure of the 
original system. Nevertheless, in such case, the Hankel singular values are given by
\begin{equation*}
 \Lambda_{QP}=\sqrt{\delta_{o}\delta_{c}}I_{n}.
\end{equation*} 
Since all the entries the matrix $\Lambda_{QP}$ are equal, the criterion of truncating the states related to the smallest singular values is impractical and further information is required to decide which states can be removed. To deal
with this situation, we proceed as follows:
\begin{itemize}
 \item We fix $\breve P$ as in \eqref{PRLC}.
 \item We propose $\beta$ and a diagonal matrix $\Gamma_c$ sucht that $T$, defined in \eqref{Tsel}, solves the LMI \eqref{LMIcontr}.
 \item We look for a matrix $\Lambda_{QT}$ that solves the inequality \eqref{cexeo}. We stress the fact that, in this case
 $Q$, given by \eqref{Qe}, is a diagonal matrix. 
 \item We fix $\alpha=\beta$ and we look for a diagonal matrix $\Gamma_{o}$ such that $S$, given by \eqref{symS}, solves the LMI \eqref{LMIobs}.
 \item We find a transformation that balances the system. Then, we truncate the system to obtain the reduced order model.
\end{itemize}

To illustrate the methodology, we consider the values in Table \ref{T:rlc}. Then, we propose\footnote{A large $\delta_{c}$ is translated in large values of the entries of $\breve P$ which can,
potentially, produce large singular values.} $\delta_{c}=0.11$ in \eqref{PRLC}. Hence, the design parameters
\begin{equation}
 \begin{array}{rcl}
  \Gamma_{c}&=&-\diag \{14, 4.9, 3.7, 0, 0, 190,  600, 350, 3.9, 10\}\\
  \beta&=&\alpha=5\times 10^{8}
 \end{array}
\end{equation} 
ensure that
\begin{equation}
 \begin{array}{rcl}
 T&=&\diag \{0.08, 0.18, 0.06, 121.21, 38.68,\\&& 0.02,  0.04, 0.07, 29.66, 64.52\}\times10^{-4}, \label{TRLC}
 \end{array}
\end{equation} 
solves \eqref{LMIcontr}. Now, using Matlab, we solve the inequality \eqref{cexeo} to obtain 
 \begin{equation}
 \begin{array}{rcl}
 \Lambda_{QT}&=&\diag \{5.89, 5.85, 6.23, 6.56, 6.83,\\&& 6.93,  6.5, 6.63, 5.84, 5.61\}\times10^{3}\\[.2cm]
  Q&=&\diag \{0.39, 0.78, 0.21, 414.89, 134.25,\\&&   0.09, 0.19, 0.28, 101.3, 202.93\}\times10^{3}.
 \end{array}\label{Qrlc}
\end{equation}
Moreover, we propose
\begin{equation}
 \begin{array}{rcl}
  \Gamma_{o}&=&\diag \{0, 0, 0, 0.2, 0.1, 0,  0, 0, 1, 5\}\times10^{12}. \label{gamorlc}
 \end{array}
\end{equation} 
Hence, we replace $Q$, given in \eqref{Qrlc}, and \eqref{gamorlc} in \eqref{symS}. Accordingly,
\begin{equation}
 \begin{array}{rcl}
  S&=&\diag \{0.08, 0.16, 0.04, 82.9, 26.81,\\&& 0.02,  0.04, 0.06, 19.87, 38.68\}\times10^{-5}, \label{SRLC}
 \end{array}
\end{equation} 
which solves the LMI \eqref{LMIobs}.\\
Note that 
\begin{equation*}
 H T^{-1}S = S T^{-1}H.
\end{equation*} 
Hence, it follows from Theorem \ref{th:diag} that there exist a transformation $W$ that balances the system and preserve the PH structure. Moreover,
the matrices $H, T$ and $S$ are diagonal. As a result, $W$ is a block diagonal matrix, thus, we can express
 the matrices $W$ and $\Lambda_{ST}$ as follows
\begin{equation}
 \begin{array}{rcl}
  W&=&\diag\{W_{{1}},  W_{{2}}\} \\[.2cm]
 \Lambda_{ST}&=&\diag\{\Lambda_{ST_{1}}, \Lambda_{ST_{2}}\}\\[.2cm]
 \Lambda_{ST_{i}}&=&\diag\{\sigma_{i_{1}},\cdots,\sigma_{i_{5}}\}, \; i=1,2,
 \end{array}
\end{equation} 
where  
\begin{equation}
 \begin{array}{rcl}
 W_{1}&=&\begin{bmatrix}
        629.3 & 0 & 0 & 0 & 0 \\
        0 & 433 & 0 & 0 & 0 \\ 
        0 & 0 & 807.2 & 0 & 0 \\
        0 & 0 & 0 & 0 & 17.8 \\ 
        0 & 0 & 0 & 31.3 & 0
       \end{bmatrix} \\[1cm]
 W_{2}&=& \diag\{1332, 892.1,714.2,36.1,25.2\}\\[0.1cm]
  \Lambda_{ST_{1}}&=&\diag\{0.31, 0.29, 0.28, 0.26, 0.26\} \\[0.1cm]
  \Lambda_{ST_{2}}&=&\diag\{0.31,0.3,0.29,0.26,0.24\}.
 \end{array}
\end{equation} 
The criterion to choose the parameters $\Gamma_{c}, \Gamma_{o}$, and $\beta$ differs from the example studied in Section \ref{sec:mecex}. In this case, we want to have a significant contrast in the entries of $\Lambda_{ST}$ to have information about which states can be truncated without affecting the response of the reduced-order system significantly. The mentioned parameters were selected by trial and error for the sake of illustration, but these might be computed by solving an optimization problem.\\[0.2cm]
At this point, we make three observations regarding the preservation of the RLC structure:
\begin{itemize}
 \item [(i)] As mentioned before, to preserve the RLC structure it is necessary to ensure that $W$ is a block diagonal matrix.
 \item [(ii)] We are truncating the states related to
the entries of $\Lambda_{ST}$ in pairs, that is, one state related to one element of $\Lambda_{ST_{1}}$ and one state related to one entry from $\Lambda_{ST_{2}}$. 
The physical interpretation of this condition is that we are removing the same number of inductors and capacitors.
\item [(iii)] By fixing $\Gamma_o$ and $\Gamma_c$ different from zero, we ensure that the entries of $\Lambda_{ST}$ are different. Then, we can apply the criterion of truncating the states 
related to the smallest entries of each submatrix $\Lambda_{ST_{i}}$.
\end{itemize}
For illustration purposes, we truncate the states related to $\sigma_{i_{4}},\sigma_{i_{5}}$. In such a case, the reduced order model admits a PH representation with 
\begin{equation}
 \begin{array}{rcl}
  R_{r}&=&\diag\{R_{C_{r}}^{-1}, R_{L_{r}} \}, \; \; J_{r}=\begin{bmatrix}
           \mathbf{0}_{3\times 3} & J_{1_{r}} \\ J_{1_{r}}^{\top} &
          \mathbf{0}_{3\times 3}
          \end{bmatrix},\\[0.3cm]
H_{r}&=&\diag\left\lbrace \frac{1}{C_{1_{r}}}, \frac{1}{C_{2_{r}}},\frac{1}{C_{3_{r}}}, \frac{1}{L_{1_{r}}}, \frac{1}{L_{2_{r}}}, \frac{1}{L_{3_{r}}}\right\rbrace,\\[0.2cm]
R_{C_{r}}&:=&\diag\left\lbrace R_{C_{1_{r}}}, R_{C_{2_{r}}}, R_{C_{3_{r}}}  \right\rbrace, \\[0.2cm]
R_{L_{r}}&:=&\diag\left\lbrace R_{L_{1_{r}}}, R_{L_{2_{r}}}, R_{L_{3_{r}}}  \right\rbrace, \\[0.2cm]
J_{1_{r}}&=&\begin{bmatrix}
             1 & -\gamma_{2} & 0 \\ 0 & 1 & -\gamma_{3} \\
             0 & 0 & 1
            \end{bmatrix}, \; \; B_{r}=\begin{bmatrix}
         \mathbf{0}_{3}\\ \gamma_{1} \\ \mathbf{0}_{2}
        \end{bmatrix}
\end{array}
\end{equation} 
and the values given in Table \ref{T:redrlc}. The error bound is given by

\begin{equation}
 \lVert \varSigma - \varSigma_{r} \rVert_{\infty}\leq 2.06,
\end{equation} 

and the reduced order model admits the RLC realization depicted in Fig. \ref{fig:realization}, where the states $\hat{x}_{i}$ represent the charges in the capacitors 
and $\hat{x}_{i+3}$ denote the fluxes in the inductors for $i=1,2,3$.

\begin{figure*}
 \centering
 \captionsetup{justification=centering}
 \includegraphics[width=\textwidth]{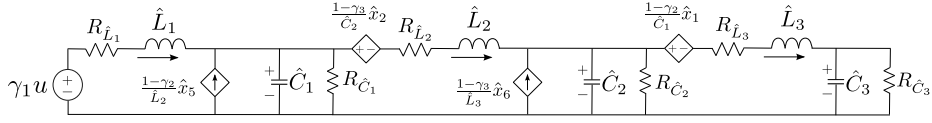}
 \vspace{.2cm}
  \caption{Reduced RLC network}
 \label{fig:realization}
\end{figure*}

\begin{table}[h]
\caption{\\Parameters of the reduced RLC circuit}
 \begin{multicols*}{3}
%\begin{center}
\begin{tabular}{|l|l|} \hline
$\gamma_{1}$ & $0.69\times10^{-4}$\\ \hline
$\gamma_{2}$ & $1.01$\\ \hline
$\gamma_{3}$ & $1.53$ \\ \hline
\end{tabular}
%\end{center}
%\begin{center}
\begin{tabular}{|l|l|} \hline
$R_{C_{1_{r}}}$ & $127.55$\\ \hline
$R_{C_{2_{r}}}$ & $485.34$\\ \hline
$R_{C_{3_{r}}}$ & $373.01$\\ \hline
$R_{L_{1_{r}}}$ & $2.22$ \\ \hline
$R_{L_{2_{r}}}$ & $1.89$\\ \hline
$R_{L_{3_{r}}}$ & $2.49$ \\\hline
\end{tabular}
%\end{center}
%\begin{center}
\begin{tabular}{|l|l|} \hline
$C_{1_{r}}$ & $4.66\times10^{-3}$\\ \hline
$C_{2_{r}}$ & $2.06\times10^{-3}$\\ \hline
$C_{3_{r}}$ & $2.92\times10^{-3}$\\ \hline
$L_{1_{r}}$ & $4.72\times10^{-3}$ \\ \hline
$L_{2_{r}}$ & $2.09\times10^{-3}$\\ \hline
$L_{3_{r}}$ & $3.05\times10^{-3}$ \\\hline
\end{tabular}
%\end{center}
\end{multicols*}
\label{T:redrlc}
\end{table}

{\bf Simulation results}\\
We carry out simulations to compare the behavior of the original system with:
\begin{itemize}
 \item A reduced order system obtained via generalized balancing, where the generalized Gramians are chosen as in \eqref{PRLC} and $\delta_{o}=\delta_{c}$.
 \item The reduced order system obtained through extended balanced truncation, with Gramians \eqref{TRLC} and \eqref{SRLC}.
\end{itemize}
A first set of simulations is performed considering that the systems start at rest and the input signal depicted in Fig. \ref{fig:u1rlc}. Figure \ref{fig:ys1rlc} shows
the outputs of the systems, where $y$ is the output of the balanced system, $y_{G}$ represents the output of the system obtained via generalized balancing, 
and $y_{E}$ corresponds to the output of the system obtained via extended balanced truncation. In Fig. \ref{fig:ys1rlc} we can observe that the difference between $y_{E}$ and $y$ 
is rather small, while the output $y_{G}$ is considerably different from the output of the balanced system, this can be the result of truncating states without any justification in 
the generalized balanced truncation case. Figure \ref{fig:eg1} shows the plot of the difference $y-y_{G}$ and Fig. \ref{fig:ee1} depicts the difference $y-y_{E}$, if we compare both plots,
we corroborate that---note that the scales of the plots are different---the error $y-y_{G}$ is noticeably bigger than the error $y-y_{E}$.\\[0.2cm]
A second set of simulations is carried out considering the input shown in Fig. \ref{fig:u2rlc} and the systems starting at rest. Figure \ref{fig:ys2rlc} shows the outputs of the systems, where it is clear that the output
$y_{G}$, plotted in blue, is totally different from the output of the balanced system. As we discussed above, the reason for this difference is the lack of a criterion to truncate
the states in the generalized balancing approach. On the other hand, we observe in Fig. \ref{fig:ys2rlc} that the plot of $y_{E}$ approximates the behavior of $y$. Hence we conclude that,
in this example, the matrices $\Gamma_c$ and $\Gamma_o$ can be exploited to reduce the error bound while preserving the physical interpretation of the original system.
% The results of these simulations are shown in ,\ref{fig:eg2} and \ref{fig:ee2}. 
% In both set of simulations, it is clear that the magnitude of the error is much smaller using extended balanced truncation approach.
\begin{figure}
 \centering
 \includegraphics[scale=.34]{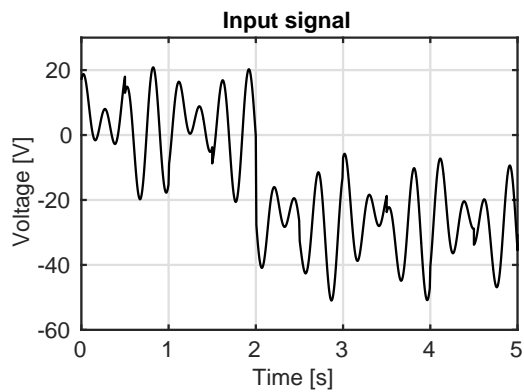}
 % umec.eps: 0x0 px, 300dpi, 0.00x0.00 cm, bb=
 \caption{Input signal $u$.}
 \label{fig:u1rlc}
\end{figure}
\begin{figure}
 \centering
 \includegraphics[scale=.34]{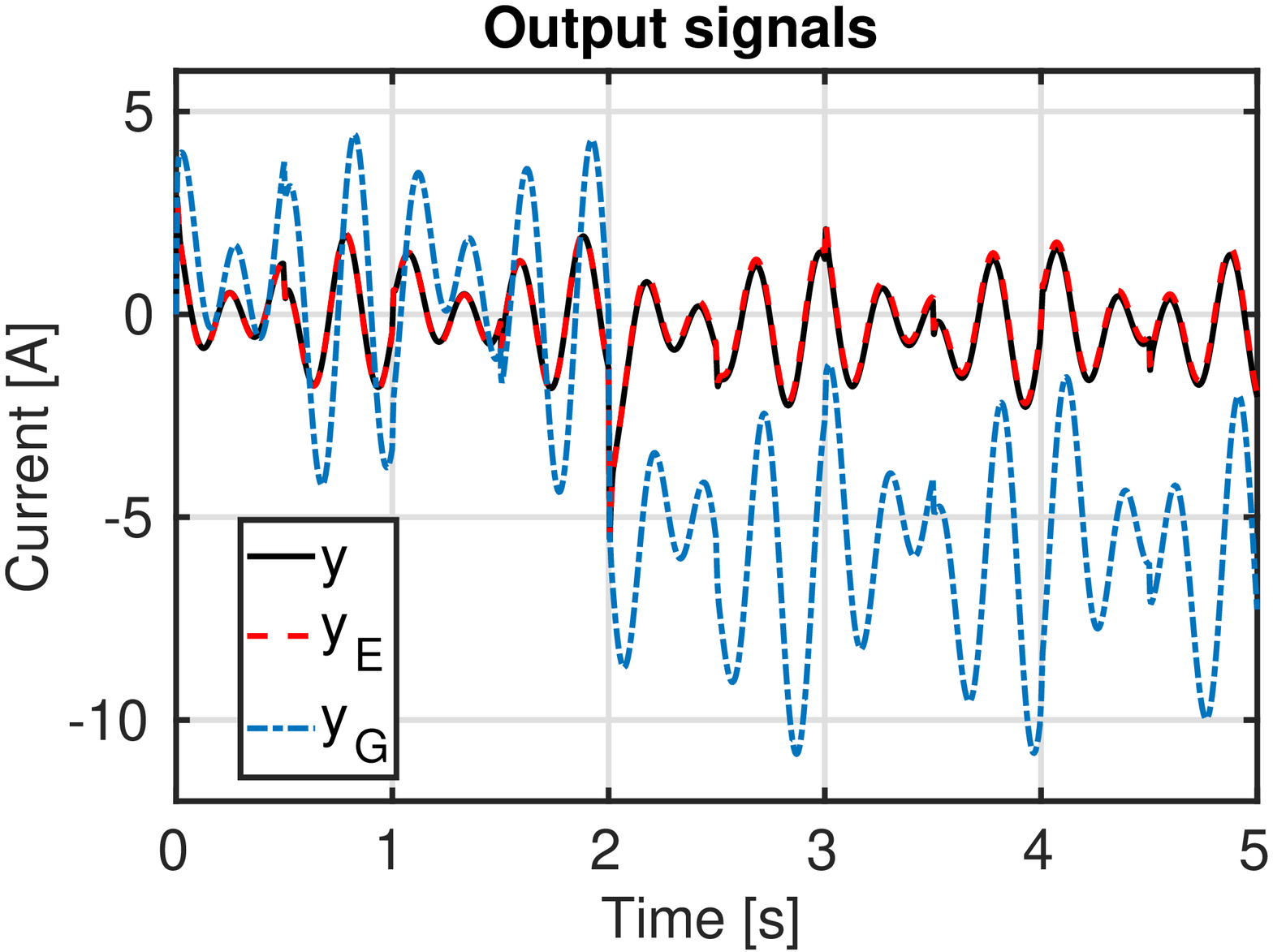}
 % ymec.eps: 0x0 px, 300dpi, 0.00x0.00 cm, bb=
 \caption{Plot of the different outputs.}
 \label{fig:ys1rlc}
\end{figure}
\begin{figure}
 \centering
 \includegraphics[scale=.34]{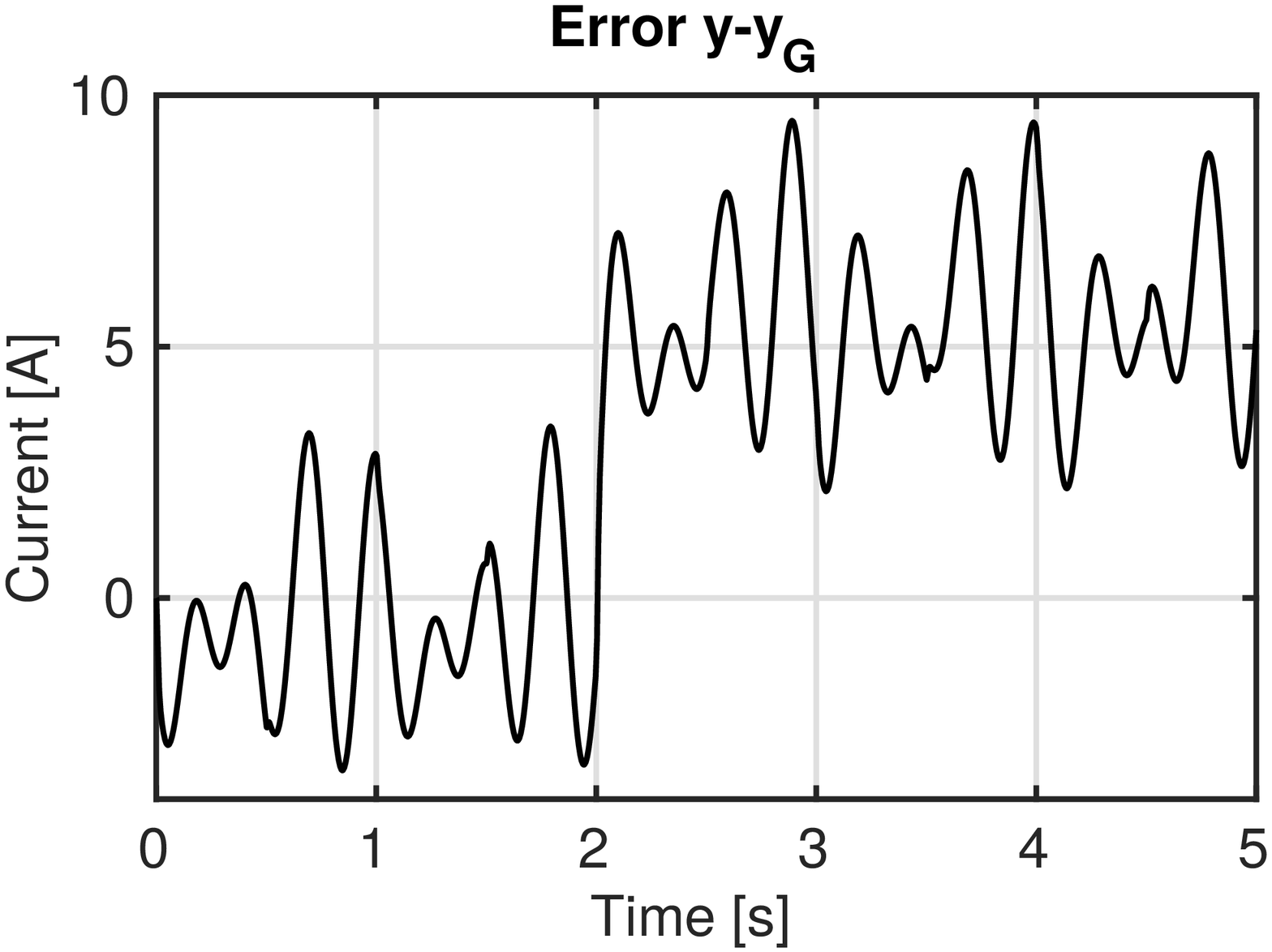}
 % umec.eps: 0x0 px, 300dpi, 0.00x0.00 cm, bb=
 \caption{Plot of the error $y-y_{G}$.}
 \label{fig:eg1}
\end{figure}
\begin{figure}
 \centering
 \includegraphics[scale=.34]{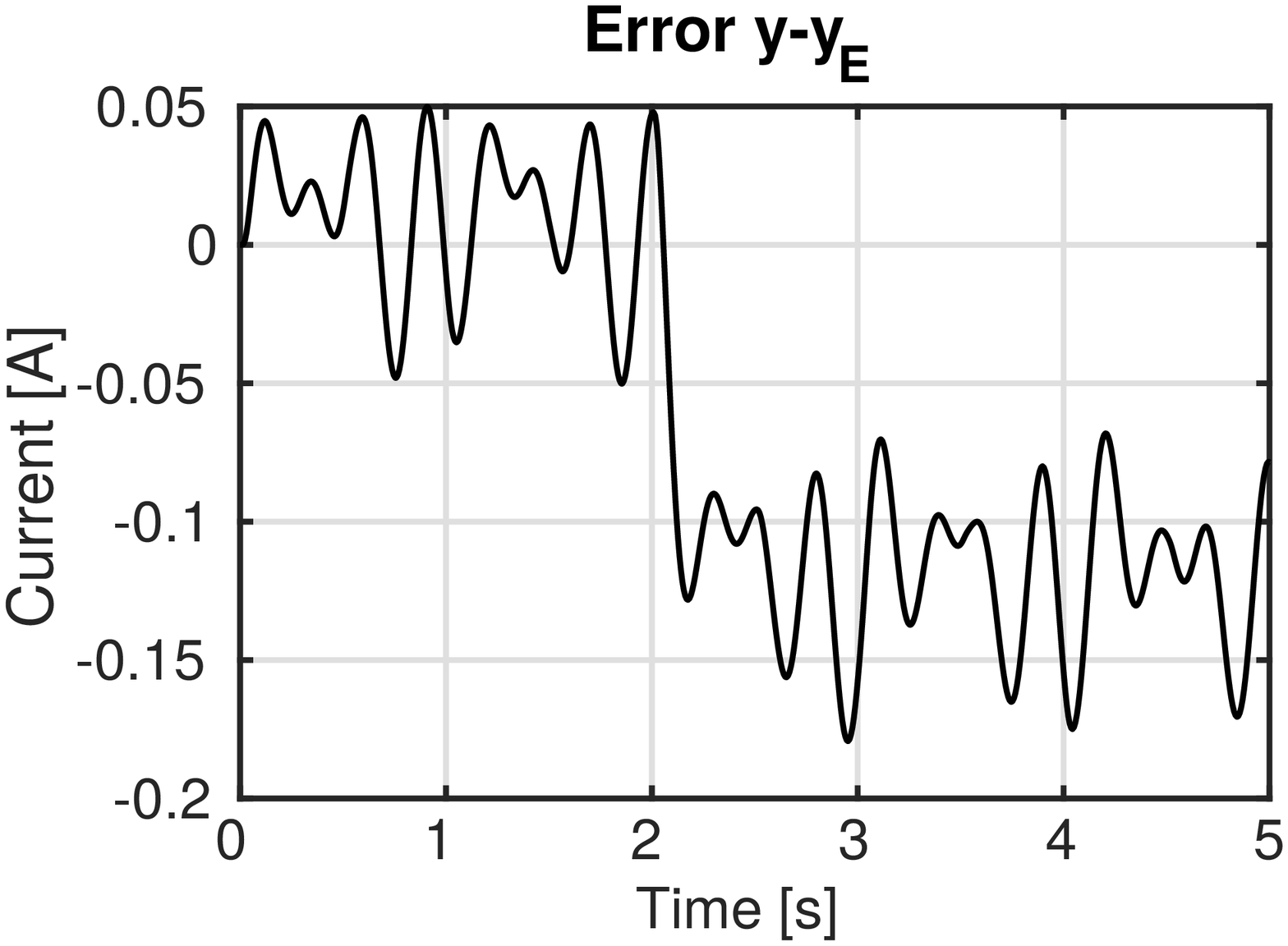}
 % umec.eps: 0x0 px, 300dpi, 0.00x0.00 cm, bb=
 \caption{Plot of the error $y-y_{E}$.}
 \label{fig:ee1}
\end{figure}

\begin{figure}
 \centering
 \includegraphics[scale=.34]{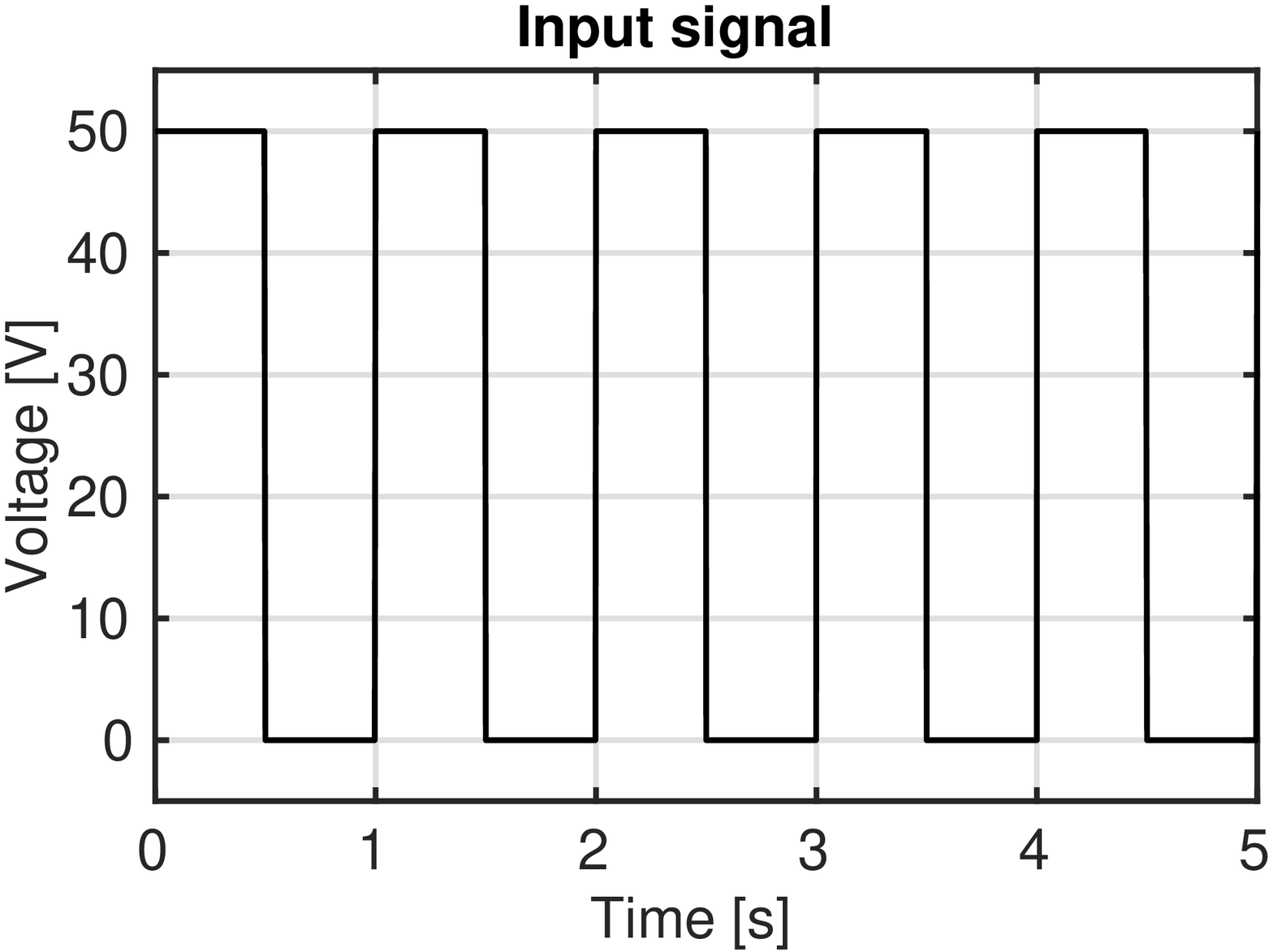}
 % umec.eps: 0x0 px, 300dpi, 0.00x0.00 cm, bb=
 \caption{Input signal $u$.}
 \label{fig:u2rlc}
\end{figure}
\begin{figure}
 \centering
 \includegraphics[scale=.34]{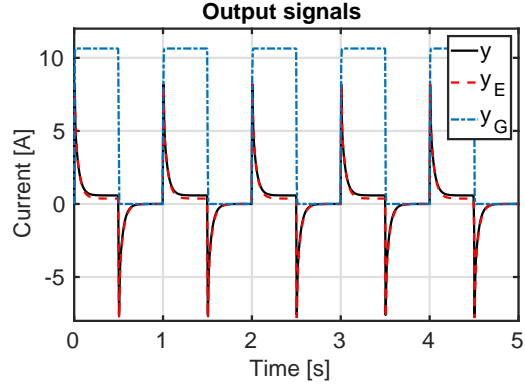}
 % ymec.eps: 0x0 px, 300dpi, 0.00x0.00 cm, bb=
 \caption{Plot of the different outputs.}
 \label{fig:ys2rlc}
\end{figure}
% \begin{figure}
%  \centering
%  \includegraphics[scale=.34]{eg2_rlc.eps}
%  % umec.eps: 0x0 px, 300dpi, 0.00x0.00 cm, bb=
%  \caption{Plot of the error $y-y_{G}$.}
%  \label{fig:eg2}
% \end{figure}
% \begin{figure}
%  \centering
%  \includegraphics[scale=.34]{ee2_rlc.eps}
%  % umec.eps: 0x0 px, 300dpi, 0.00x0.00 cm, bb=
%  \caption{Plot of the error $y-y_{E}$.}
%  \label{fig:ee2}
% \end{figure}
%%%%%%%%%%%%%%%%%%%%%%%%%%%%
\section{Concluding remarks}
\label{sec:cr}
In this paper we have provided sufficient conditions to ensure the existence of extended Gramians that are suitable to compute an error bound. Additionally, we have formulated an approach to preserve the PH structure
for the truncated system by using generalized and extended Gramians. Furthermore, we have shown that the extended balancing is a versatile tool that can be used to obtain a smaller error bound or to preserve
some particular structures, such as, an RLC structure.
%%%%%%%%%%%%%%%%%%%%%%%%%%%%
\appendix\label{app:mech}
% \subsection{Mechanical system}\label{app:mech}
 The matrices related to the mechanical example are listed below
 \scriptsize
\begin{equation}\label{Ptmec}
 \begin{array}{l}
  \breve{P}=\left[ \begin{array}{ccccc}
                    0.97 & 0.37 & 0.35 & 0.29 & 0.15 \\
                    0.37 & 0.46 & 0.44 & 0.39 & 0.26 \\
                    0.35 & 0.44 &  0.43 & 0.38 & 0.26 \\
                    0.29 & 0.39 &  0.38 & 0.35  & 0.24 \\
                    0.15 & 0.26 &  0.26 &  0.24 & 0.18 \\
                    0 & -0.13 & -0.06 & 0.03  &  0.08 \\
                    0.04 & 0 & 0  & 0  & 0 \\
                    0.15 & 0.03 &  0 &  0.04 & 0.06 \\
                    -0.05 & -0.01 & -0.02  & 0  &  0.02 \\
                    -0.07 & -0.01 & -0.02  &  -0.01 & 0 \\
                   \end{array}\right. \\ \\
\; \; \; \; \; \; \; \; \; \; \; \; \; \; \; \; \left.  \begin{array}{ccccc}
                    0 & 0.04 & 0.15 & -0.05 & -0.07 \\
                    -0.13 & 0 & 0.03 & -0.01 & -0.01 \\
                    -0.06 & 0 & 0 & -0.02 & -0.02 \\
                    0.03 & 0 & 0.04 & 0 & -0.01 \\
                    0.08 & 0 & 0.06 & 0.02 & 0 \\
                    3.77 & -0.19 & -1.56 & -0.78 & -0.55 \\
                    -0.19  & 0.04  & 0.32 & 0.15 & 0.08 \\
                    -1.56  &  0.32 & 2.52 & 1.19 & 0.63 \\
                    -0.78  & 0.15  & 1.19 & 0.58 & 0.32 \\
                    -0.55 & 0.08  & 0.63 & 0.32 & 0.18
                   \end{array}\right]
                   \end{array}
%  
%  
%  
%  
%  \begin{bmatrix}
%            0.97 & \star & \star & \star & \star & \star & \star & \star & \star & \star \\
%            0.37 & 0.46 & \star & \star & \star & \star & \star & \star & \star & \star \\
%            0.35 & 0.44 &  0.43 & \star & \star & \star & \star & \star & \star & \star \\ 
%            0.29 & 0.39 &  0.38 & 0.35  & \star & \star & \star & \star & \star & \star \\ 
%            0.15 & 0.26 &  0.26 &  0.24 & 0.18  & \star & \star & \star & \star & \star \\
%            0 &   -0.13 &  -0.06 & 0.03  &  0.08 & 3.77 &\star & \star & \star & \star \\
%            0.04 & 0 & 0  & 0  & 0  & -0.19  & 0.04  & \star & \star & \star \\
%            0.15 & 0.03 &  0 &  0.04 & 0.06  & -1.56  &  0.32 & 2.52 & \star & \star \\
%            -0.05 & -0.01 & -0.02  & 0  &  0.02 & -0.78  & 0.15  & 1.19 & 0.58 & \star \\
%            -0.07 & -0.01 & -0.02  &  -0.01 & 0  &  -0.55 & 0.08  & 0.63 & 0.32 & 0.18 
%            \end{bmatrix}
\end{equation}
\begin{equation}\label{Gamcmec}
 \begin{array}{l}
  \Gamma_c=\left[ \begin{array}{ccccc}
                    0.05 & -0.1 & -0.07 & -0.05 & -0.03 \\
                    -0.1 & 0.01 & 0 & -0.01 & 0 \\
                    -0.07 & 0 &  -0.01 & -0.02 & -0.01 \\
                    -0.05 & -0.01 & -0.02 & -0.02  & -0.01 \\
                    -0.03 & 0 &  -0.01 &  -0.01 & -0.01 \\
                    1.63 & -0.56 & -0.57 & -0.54 & -0.5 \\
                    -0.12 & 0.03 & 0.04 & 0.04 & 0.04 \\
                    -1.01 & 0.31 & 0.33 & 0.35 & 0.34 \\
                    -0.51 & 0.17 & 0.17 & 0.18 & 0.17 \\
                    -0.32 & 0.11 & 0.11 & 0.11 & 0.1 \\
                   \end{array}\right. \\ \\
\; \; \; \; \; \; \; \; \; \; \; \; \; \; \; \; \left.  \begin{array}{ccccc}
                    1.63 & -0.12 & -1.01 & -0.51 & -0.32 \\
                    -0.56 & 0.03 & 0.31 & 0.17 & 0.11 \\
                    -0.57 & 0.04 & 0.33 & 0.17 & 0.11 \\
                    -0.54 & 0.04 & 0.35 & 0.18 & 0.11 \\
                    -0.5 & 0.04 & 0.34 & 0.17 & 0.1 \\
                    -0.29 & 0.18 & 0.8 & 0.11 & -0.02 \\
                    0.18 & -0.02  & -0.11 & -0.04 & -0.02 \\
                    0.8 & -0.11 & -0.52 & -0.09 & -0.03 \\
                    0.11 &  -0.04 &  -0.09 &   0.04 & 0.04 \\
                    -0.02 &  -0.02 &   -0.03 &   0.04 &    0.04
                   \end{array}\right]
                   \end{array}
\end{equation}
%%%%%%%%%%%%%%%%%%%%%%%%%%%%

\bibliographystyle{plain}
\bibliography{bibmr}

\end{document}